\documentclass[11pt]{article}

\usepackage{macros}

\newif\ifnotes
\notestrue

\title{Linear-Time Encodable Quantum Codes near the CSS GV Bound}

\author{Rachel Yun Zhang\thanks{UC Berkeley.  \texttt{rachelyz@berkeley.edu}. This work was done under the support of the Miller Institute for Basic Research in Science, University of California Berkeley.}}

\date{\today}

\begin{document}

\maketitle

\begin{abstract}
    We construct quantum CSS codes with rate-distance tradeoff approaching the CSS GV bound, which can be encoded by a quantum circuit of logarithmic depth and a linear number of gates. Our construction is heavily inspired by a work of Brehm and Resch and can be viewed as a quantum analogue of classical \emph{repeat multiple accumulate} codes. The encoder has a particularly simple structure, consisting of a single constant depth outer quantum circuit followed by several classical accumulation and inverse-accumulation layers.
\end{abstract}

\thispagestyle{empty}
\setcounter{page}{0}
\newpage

\enlargethispage{1cm}
\tableofcontents
\pagenumbering{roman}
\newpage
\pagenumbering{arabic}

\section{Introduction}

Recently, \cite{WillsLZH26} constructed asymptotically good quantum codes that can be encoded and unencoded by quantum circuits consisting of a linear number of gates. These codes are relevant to the channel capacity setting, where a sender may perfectly encode a state before sending it over a noisy channel. Even when perfect encoding and unencoding is not guaranteed, a quantum processor can use an error-correcting code to run the encoding and unencoding operations at the logical level; low complexity encoding and unencoding is highly desirable to avoid blowup in the task being performed.

\cite{WillsLZH26}'s construction is inspired by \cite{Spielman96}'s construction of classical linear time encodable and decodable codes. While these codes have linear rate and distance, the rate-distance tradeoffs depend on the underlying expander graph and remain far from optimal. The same holds for \cite{WillsLZH26}'s construction of linear time encodable and unencodable quantum error-correcting codes.

In the 2000s, a line of work \cite{DivsalarJM98,PfisterS03,Pfister03,BazziMS09,FagnaniR08,RavazziF10,KliewerZKC08} analyzed (classical) error-correcting codes with very fast encoding, called \emph{repeat multiple accumulate (RMA)} codes. In an RMA code, the encoder starts by \emph{repeating} each of the $n/\rep$ bits of the plaintext $\rep$ times. It then performs several rounds of \emph{permute} and \emph{accumulate} on the resulting $n$ bits. In each permute and accumulate round, first a random coordinate permutation is applied, followed by an accumulation, which is a prefix sum modulo $2$. Importantly, accumulation can be implemented very quickly, in a $\log n$ depth circuit consisting of just $2n-1$ XORs~\cite{LadnerF80}. 

The distance of this very simple iterative encoder was analyzed in~\cite{FagnaniR08}, who showed that as the number $m$ of accumulation rounds grows, the distance of the resulting $RMA$ code approaches the Gilbert-Varshamov (GV) bound. The GV bound represents the best known achievable tradeoff of rate and distance in binary codes. The fact that such a simple encoder approaches it is quite beautiful.

In this work, we quantize RMA codes. Our main result is a construction of a quantum CSS code with an extremely efficient encoder, whose rate-distance tradeoff lies near the quantum CSS GV bound. By extremely efficient, we mean not only that the encoding circuit is of linear size and logarithmic depth, but the constants are also very reasonable. 

\begin{theorem}
    For any even $\rep \ge 4$ and $\eps > 0$, there is a quantum CSS code with rate $R = 1 - 2/\rep$ and also has both $X$-distance and $Z$-distance $\ge \left( \delta_{CSS-GV}(R) - \eps \right) \cdot n$. The encoding circuit acts on $(1 - \frac2\rep) n$ logical qubits and $\frac{n}{s}$ $\ket{0}$ and $\frac{n}{s}$ $\ket{+}$ ancilla qubits, has depth $O(\log n)$, and consists of $O(n)$ CNOT gates.
\end{theorem}

In the above theorem, $\delta_{GV}(R)$ represents the largest relative distance we know to be achievable by a rate $R$ CSS code. It is achieved by a CSS code with randomly sampled $X$ and $Z$ stabilizers. 

\begin{definition}[\cite{CalderbankS96}]
    $\delta_{CSS-GV}(R) = h^{-1}\left( \frac{1 - R}{2} \right) \in [0, \frac12]$. 
\end{definition}

In this work, we prove that the relative distance of our quantum code approaches $\delta_{CSS-GV}$ as the number $m$ of accumulation rounds grows, but we do not give an explicit characterization of the rate of convergence. However, in practice, it appears that distance converges quite quickly to the CSS GV bound. In \Cref{fig:numerics}, we numerically estimate the distances of our quantum code for a few code rates and number of encoding rounds. We see that after just $4$ encoding rounds, the relative distance of our code is within $0.001$ of the CSS GV bound; after $6$, it is within $10^{-7}$.

\begin{figure}[h!]
    \centering
    \begin{tabular}{c|c|c|c|c|c|c}
        $\rep$ & Code Rate & $m = 3$ & $m = 4$ & $m = 5$ & $m = 6 $& CSS GV Bound \\
        \hline
        $4$ & $0.5$ & $0.0311830$ & $0.0413154$ & $0.0416915$ & $0.0416927$ & $0.0416927$ \\
        $6$ & $0.667$ & $0.0169172$ & $0.0241912$ & $0.0246207$ & $0.0246235$ & $0.0246235$ \\
        $8$ & $0.75$ & $0.0111711$ & $0.0166938$ & $0.0171245$ & $0.0171287$ & $0.0171287$
    \end{tabular}
    \caption{Rates and numerically estimated asymptotic distances of our quantum code. The parameter $m$ describes the number of encoding rounds and controls the size of the encoding circuit, which has depth $\le 2s + 4m \log_2(n)$ and consists of $\le (4m+2) n$ CNOT gates acting on $\left( 1 - \frac2\rep \right) n$ logical qubits and $n/\rep$ $\ket{0}$ and $n/s$ $\ket{+}$ ancillas qubits.}
    \label{fig:numerics}
\end{figure}

\subsection{Overview}

A direct attempt to interpret a RMA code as a quantum code leads to horrible quantum distance. The reason is that, assuming the number of accumulation rounds is constant (as is typically the case), the dual of an RMA code contains codewords of constant weight.

Seeking to construct quickly encodable classical codes whose dual and itself both lie near the GV bound, \cite{BR26} get around this by interleaving not just accumulation rounds, but also ``derivative'' rounds. A derivative simply takes, at every coordinate, the XOR of the current and immediately preceding coordinate, treating the $0$'th coordinate as having value $0$. It is also the inverse of the accumulation (prefix-sum) operation. The idea is that, while a derivative is unlikely to increase the distance of the code, it \emph{is} likely to increase the distance of the dual code. 

\cite{BR26} thereby define a new family of codes where after the initial repetition step, permuted accumulation steps are interleaved with permuted derivative steps for $m$ rounds. They call this family of codes a $R(AD)^m$ code. They analyze this code where the repetition factor is $2$, obtaining a rate-$\frac12$ code lying near the GV bound. Importantly, the dual code can be viewed as a $R(DA)^m$ code, where after an initial $2$-repetition, permuted derivatives and permuted accumulations are interleaved. They analyze the dual and show that it too is a rate-$\frac12$ code lying near the GV bound.

Interpreted as a quantum code, \cite{BR26}'s construction is a CSS code encoding $0$ qubits. We will modify their construction to obtain a CSS code encoding a linear number of qubits. Let us view the primary $R(AD)^m$ code as the $X$ stabilizers, and the dual $R(DA)^m$ code as the $Z$ stabilizers. By taking the repetition factor to be a larger even $\rep \ge 4$, we obtain stabilizers for a CSS code $\calC$ of rate $1 - \frac2\rep$. The $Z$ and $X$ distances respectively reduce to understanding the distance of a classical $P_\rep(DA)^m$ and $P_\rep(AD)^m$ code, which we define to have an $\rep$-parity code as the outer encoder instead of the repetition code. The $\rep$-parity code encodes $(1 - \frac1\rep)n$ bits by appending a single check-sum bit to each block of $\rep-1$ bits. We remark that our analysis of the distance of $P_\rep(DA)^m$ and $P_\rep(AD)^m$ codes follows the same general outline as~\cite{BR26}, but some work is needed to get suitable bounds on the outer weight enumerator (see~\Cref{prop:initial-properties}). In particular, instead of working with \emph{exact} spectral approximations to the weight enumerator, we will only give a suitable upper bound (see~\Cref{def:varphi}), so that all the computed spectral shapes are loose upper bounds on the logarithm of the expected number of codewords of a given weight. This choice of upper bound on the outer weight enumerator cannot be too loose, however: we will need it to be tight enough to prove convergence to the CSS GV bound (see~\Cref{sec:convergence-to-GV}). 

One point to note is that this description of a CSS code by its $X$ and $Z$ stabilizer code spaces does not automatically permit a linear time encoding, even if $P_\rep(DA)^m$ and $P_\rep(AD)^m$ themselves have fast encoding algorithms. The reason is that a quantum encoding algorithm needs to coherently encode the logical information into appropriate coset superpositions, which requires a compatible choice of bases for the two nested classical code spaces, not merely efficient encoders for the spaces individually. Instead, we show that one can encode $\calC$ by first applying a \emph{quantum} outer code, then applying interleaved accumulation and derivative rounds (see~\Cref{sec:linear-time-encoder}).

We did not study error-correction/decoding complexity in this work. We remark that there is a vast literature on decoding classical RMA codes via belief propagation \cite{BenedettoDMP98,McElieceMJ98,KschischangFL01,RoumyGCV03,Mansour04,AmatR09} -- it would be interesting to extend these techniques to our quantum code.


\subsection{Related Work}

The history of RMA codes can be traced back to the early work of~\cite{DivsalarJM98}, who introduced $RA$ codes consisting of a single repetition and accumulation round. Unfortunately,~\cite{KahaleU98} showed that these single round codes do not have linear distance. This led Pfister and Siegel~\cite{PfisterS03} to first study $RA^m$ codes with multiple permute and accumulate rounds. They numerically estimated the expected number of low weight codewords, and found that with two or more rounds, $RA^m$ codes appeared to have positive relative distance quite close to the GV bound. This was extended in Chapter 3 of Pfister's thesis~\cite{Pfister03}, which contains the first formal proof that $RA^m$ codes have positive relative distance for $m \ge 2$, as well as precise numerical estimates of the minimum distance attained by $RA^m$ codes of varying rates and number of rounds. At this point, Pfister's proof only showed small positive relative distance, still far from the simulated distances approaching the GV bound.

It was a few years before~\cite{KliewerZKC08} used Stirling's approximation to simplify the complicated IOWEF formulas, leading to the first formal proof of the distances in Pfister's numerical estimates. This work only showed that as $n$ grows, the $RA^m$ codes were increasingly likely to have the described minimum distances, but did not dictate how quickly this convergence happened. \cite{FagnaniR08} first showed that the probability an $RA^m$ code has the prescribed distance grows as $1 - 1/\poly(n)$. Moreover, they prove that this minimum distance converges to the GV bound as $m$ grows (previously, this was conjectured, but not formally proven). Much more recently,~\cite{BrehmCFRRZ25} characterize this probability more precisely, proving that we can expect $RA^m$ codes to have good distance for concrete block lengths $n$. 

We work in our paper with \emph{generalized} RMA codes, where the first outer encoder is a parity encoder rather than a repetition encoder. Repeat accumulate codes with a different outer encoder were studied by~\cite{Pfister03,AmatB09,RavazziF10}. Pfister studied in his thesis \cite{Pfister03} the case where the outer encoder is a parity code. The work of~\cite{AmatB09} studied repeat accumulate codes with a Hamming outer encoder. \cite{RavazziF10} generalized this and proved distance approaching the GV bound under general assumptions on the outer encoder and the interleavers. 

There have been some attempts to import these classical techniques to the quantum front. \cite{WangBZW09} construct \emph{entanglement assisted} quantum codes using classical quaternary irregular repeat-accumulate codes as an underlying graphical architecture, where the sender and receiver are assumed to pre-share some entangled bits. \cite{PoulinTO09}, using ideas of recursive encoding, developed \emph{quantum serial turbo codes} in the standard setting without pre-shared entanglement. They found a fundamental difference in the quantum setting compared to classical: for ordinary stabilizer convolutional codes, the recursive encoder is necessarily \emph{catastrophic}, which means roughly that small errors are undetectably propagated into large logical errors by the inverse encoder. \cite{HoushmandW13} later gave a proof of the same limitation. They therefore instantiate their code using nonrecursive, noncatastrophic convolutional constituents, resulting in a code with bounded minimum distance. \cite{WHB14} found entanglement-assistance bypasses this limitation and constructed an entanglement-assisted quantum convolutional encoder with improved iterative-decoding on the depolarizing channel. We remark that in all these works above, the focus was on decoders against randomized noise (such as the depolarizing channel). The work of~\cite{Abbara11} analyzes the minimum distance of a quantum turbo-code, though they only showed small polynomial distance for their particular ensemble. As far as we know, we are the first to prove good (in fact, near the CSS GV bound) quantum distance of a quantum code with an iterated encoder.


\section{The Code}

We will build our code by first defining an ``outer code'' $\calC_0$, then amplifying its distance by applying a sequence of unitaries to obtain the final code $\calC$.

\paragraph{The Outer Code.} 

Let $\rep \ge 4$ be an even integer, and let $P_\rep$ denote the (classical) parity code of block length $\rep$: that is, it consists of length $\rep$ binary strings with an even number of $1$'s. Note that $P_\rep$ has dimension $\rep-1$ as the first $\rep-1$ bits fix the last, and its dual code is $R_\rep$, the $\rep$-repetition code consisting of just two codewords, $0^\rep$ and $1^\rep$. 

The code $\widetilde{P}_\rep := P_\rep^{n/\rep}$ is an $n$-bit code where codewords consist of $x_1 || x_2 || \dots || x_{n/\rep}$ where each $x_i \in P_\rep$. It has rate $R = 1 - 1/\rep$, and its dual code is $\widetilde{R}_\rep := R_\rep^{n/\rep}$ consisting of $n$-bit codewords $y_1 || y_2 || \dots || y_{n/\rep}$ where each $y_i \in R_\rep$. Since $\rep$ is even, we have that $\widetilde{R}_\rep^\top \widetilde{R}_\rep = 0$. In particular, $S_x^0 = \widetilde{R}_\rep$, $S_z^0 = \widetilde{R}_\rep$ defines a CSS code $\calC_0$ on $n$ qubits with $k = (1 - \frac2\rep) \cdot n$ logical qubits. 

\paragraph{Amplifier Operators.}

We define two (classical) operators on $n$ bits. The first is the ``accumulation'' operator, defined as
\[
    A = \begin{bmatrix}
        1 & 0 & 0 & \cdots & 0 \\
        1 & 1 & 0 & \cdots & 0 \\
        1 & 1 & 1 & \cdots & 0 \\
        \vdots & \vdots & \vdots & \ddots & \vdots \\
        1 & 1 & 1 & \cdots & 1
    \end{bmatrix}.
\]
The second classical operator is the ``derivative'' operator, defined as
\[
    D = \begin{bmatrix}
        1 & 0 & 0 & 0 & \cdots & 0 & 0 \\
        1 & 1 & 0 & 0 & \cdots & 0 & 0 \\
        0 & 1 & 1 & 0 & \cdots & 0 & 0 \\
        0 & 0 & 1 & 1 & \cdots & 0 & 0 \\
        \vdots & \vdots & \vdots & \vdots & \ddots & \vdots & \vdots \\
        0 & 0 & 0 & 0 & \cdots & 1 & 1
    \end{bmatrix}.
\]
One can check that $D = A^{-1}$.

\paragraph{The Amplified Code.}

Let $\Pi_1, \dots, \Pi_{2m}$ denote $2m$ uniformly random permutations. Define the classical operator
\[
    \op = D \circ \Pi_{2m} \circ A \circ \Pi_{2m-1} \circ \dots \circ D \circ \Pi_4 \circ A \circ \Pi_3 \circ D \circ \Pi_2 \circ A \circ \Pi_1 \in \GL_n(\bbF_2)
\]
which iteratively applies a random permutation then an accumulation, then a random permutation then a derivative, and repeats $m$ times. 
Let $U_\op$ denote the (quantum) unitary that implements $\op$. Our amplified code is defined as $\calC = U_\op \calC_0$, that is, it consists of the quantum states obtained by applying $U_\op$ to some state in $\calC_0$. 

\paragraph{Stabilizers.}

Note that
\begin{align*}
    U_\op X(a) U_\op^{-1} &= X(\op a) \\
    U_\op Z(b) U_\op^{-1} &= Z(\op^{-\top} b),
\end{align*}
therefore $\calC$ is stabilized by $X$-stabilizers $S_x = \op S_x^0$ and by $Z$-stabilizers $S_z = \op^{-\top} S_z^0$. This is a complete set of stabilizers.

\subsection{Linear Time Encoder} \label{sec:linear-time-encoder}

An encoder for $\calC$ begins with an arbitrary $\frac{\rep - 2}{\rep} \cdot n$-qubit state, applies an encoder for the outer code $\calC_0$, and then iteratively applies unitaries implementing each of $\Pi_i$ and $A$ or $D$. Applying $\Pi_i$ can be viewed as a reordering of the $n$ qubits, so it does not contribute to the total gate count. Therefore, let us describe linear-size circuits for encoding $\calC_0$, and for applying $A$ and $D$. 

\paragraph{Encoder for $\calC_0$.}

An explicit encoder for $\calC_0$ can be given as follows. For $x = x_1 \dots x_{\rep-2} \in \{ 0, 1 \}^{\rep-2}$, let $x_{\rep-1} := \oplus_{i=1}^{\rep-2} x_i$, and define $E_\rep \ket{x} \propto \ket{x, x_{\rep-1}, 0} + \ket{(x, x_{\rep-1}, 0) \oplus 1^\rep}$. $E_\rep$ can be implemented by the below (constant-size) circuit consisting of $2s-3$ CNOTs:
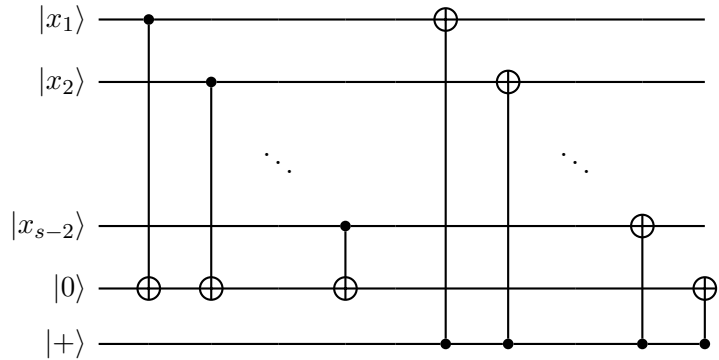
\begin{center}
\begin{quantikz}
    \lstick{$\ket{x_1}$} \qw & \ctrl{4} & \qw & \qw & \qw & \qw & \targ{} & \qw & \qw & \qw & \qw \\
    \lstick{$\ket{x_2}$} \qw & \qw & \ctrl{3} & \qw & \qw & \qw & \qw & \targ{} & \qw & \qw & \qw \\
    & & & \ddots & & & & & \ddots \\
    \lstick{$\ket{x_{\rep-2}}$} \qw & \qw & \qw & \qw & \ctrl{1} & \qw & \qw & \qw & \qw & \targ{} & \qw \\
    \lstick{$\ket{0}$} \qw & \targ{} & \targ{} & \qw & \targ{} & \qw & \qw & \qw & \qw & \qw & \targ{} \\
    \lstick{$\ket{+}$} \qw & \qw & \qw & \qw & \qw & \qw & \ctrl{-5} & \ctrl{-4} & \qw & \ctrl{-2} & \ctrl{-1} 
\end{quantikz}
\end{center}
One can check that $E_\rep \ket{x}$ is stabilized by $X^{\otimes \rep}$ and by $Z^{\otimes \rep}$. Then, we have that $E_\rep^{\otimes n/\rep} \ket{\psi}$ is stabilized by $X(S_x^0)$ and $Z(S_z^0)$, i.e. $E_\rep^{\otimes n/\rep}$ is an encoder for $\calC_0$. The separate $\rep$-blocks can be encoded in parallel, resulting in an encoding circuit of depth $2\rep - 3$ consisting of $\frac{2s-3}{s} \cdot n$ CNOTs.

\paragraph{Accumulation.} 

Next, $A$ simply computes the running prefix-sums, which~\cite{LadnerF80} showed can be computed in $\log(n)$ depth and $O(n)$ gates. We describe this circuit $\PREFIX(n)$ below.

$\PREFIX(n)$ will be constructed recursively. Suppose we have a circuit $\PREFIX(\lfloor n/2 \rfloor)$ which takes as input $\lfloor n/2 \rfloor$ bits and in place (and reversibly) outputs its prefix sum. Then, we can construct the circuit $\PREFIX(n)$ as follows: 
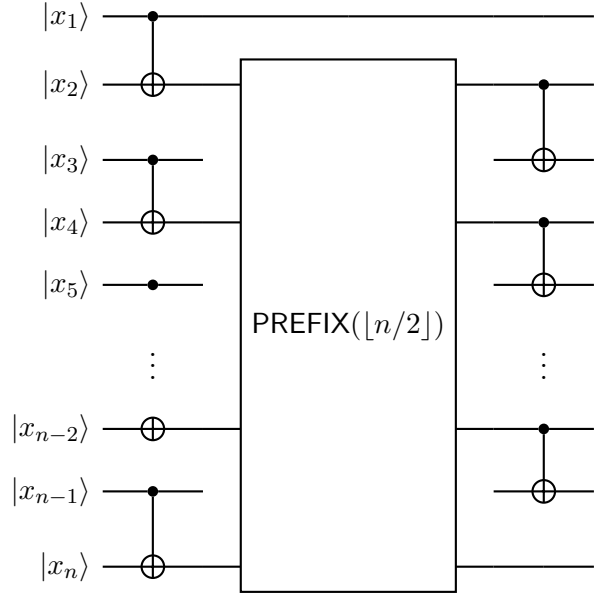
\begin{center}
\begin{quantikz}
    \lstick{$\ket{x_1}$} \qw & \ctrl{1} & \qw & \qw & \qw & \qw & \qw \\
    \lstick{$\ket{x_2}$} \qw & \targ{} & \qw & \gate[wires=8,nwires={2,4,5,7}]{\PREFIX(\lfloor n/2 \rfloor)} & \qw & \ctrl{1} & \qw \\
    \lstick{$\ket{x_3}$} \qw & \ctrl{1} & \qw &  &  & \targ{} & \qw  \\
    \lstick{$\ket{x_4}$} \qw & \targ{} & \qw &  & \qw & \ctrl{1} & \qw \\
    \lstick{$\ket{x_5}$} \qw & \ctrl{0} & \qw &  &  & \targ{} & \qw \\
    & \vdots & & & & \vdots  \\
    \lstick{$\ket{x_{n-2}}$} \qw & \targ{} & \qw &  & \qw & \ctrl{1} & \qw  \\
    \lstick{$\ket{x_{n-1}}$} \qw & \ctrl{1} & \qw &  &  & \targ{} & \qw  \\
    \lstick{$\ket{x_n}$} \qw & \targ{} & \qw &  & \qw & \qw & \qw  \\
\end{quantikz}
\end{center}
In the above diagram, the $\PREFIX(\lfloor n/2 \rfloor)$ gate only touches the even wires. (The gap for the odd wires is meant to illustrate this: we did not know how else to illustrate this in a circuit diagram.) Also, we drew the circuit for even $n$: the odd case is similar.

After the first set of $\lfloor n/2 \rfloor$ CNOTs, the even wires contain $\ket{x_{2i} + x_{2i-1}}$ while the odd wires contain their original value $\ket{x_{2i-1}}$. After the $\PREFIX(\lfloor n/2 \rfloor)$ gate, the even wires now contain the prefixes $\ket{\sum_{j \le 2i} x_j}$. Finally, the last layer of $\lceil n/2 \rceil - 1$ CNOTs add each even prefix to the following odd wire, so that the odd wires also contain the prefix sums. 

Altogether, this circuit has depth $2\log_2(n)$, and consists of $\le 2n - \log_2(n) \le 2n$ CNOT gates. 

\paragraph{Derivative.}

Finally, since $D$ is the inverse of $A$, it can be implemented by the circuit $\PREFIX(n)$ run backwards. 


\paragraph{Total.}

Therefore, $\calC$ can be encoded by a circuit of depth $2\rep - 3 + 4m \lceil \log_2(n) \rceil$, consisting of $\left( 4m + \frac{2\rep-3}{\rep} \right) \cdot n$ CNOTs. 

\subsection{Reducing Quantum Distance to Multiple Interleaved Codes }

Recall that $\calC$ is defined by $X$-stabilizers $S_x = \op S_x^0$ and by $Z$-stabilizers $S_z = \op^{-\top} S_z^0$. 

Let us begin with the $X$-distance. The $X$-distance of $\calC$ is lower bounded by the (classical) distance of the classical code $C_x = S_z^\perp$, which consists of all strings that arise as $\op x_0$ for $x_0 \in (S_z^0)^\perp = \widetilde{P}_\rep$. Recall that we defined $\op = D \circ \Pi_{2m} \circ A \circ \Pi_{2m-1} \circ \dots \circ D \circ \Pi_2 \circ A \circ \Pi_1$. So an encoding procedure for $C_x$ is to start with a string in $\widetilde{P}_\rep$ and then repeatedly apply a randomly permuted accumulation and a randomly permuted derivative $m$ times. We wish to analyze the distance of such a code, which we call $P_\rep(AD)^m$ code.

The $Z$-distance will reduce to a very similar classical code. The $Z$-distance of $\calC$ is lower bounded by the distance of the classical code $C_z = S_x^\perp$, which consists of all strings obtained as $\op^{-\top} z_0$ for $z_0 \in (S_x^0)^\perp = \widetilde{P}_\rep$. Using that $D = A^{-1}$, we can compute 
$\op^{-\top} = A^\top \circ \Pi_{2m}^{-\top} \circ D^\top \circ \Pi_{2m-1}^{-\top} \circ \dots \circ A^\top \circ \Pi_2^{-\top} \circ D^\top \circ \Pi_1^{-\top}$. But note that 
\[
    A^\top = \begin{bmatrix}
        1 & 1 & 1 & \cdots & 1 \\
        0 & 1 & 1 & \cdots & 1 \\
        0 & 0 & 1 & \cdots & 1 \\
        \vdots & \vdots & \vdots & \ddots & \vdots \\
        0 & 0 & 0 & \cdots & 1
    \end{bmatrix} 
    = FAF
\]
is simply another accumulator matrix where accumulation begins with the last bit and goes backwards, and 
\[
    D^\top = \begin{bmatrix}
        1 & 1 & 0 & 0 & \cdots & 0 & 0 \\
        0 & 1 & 1 & 0 & \cdots & 0 & 0 \\
        0 & 0 & 1 & 1 & \cdots & 0 & 0 \\
        0 & 0 & 0 & 1 & \cdots & 0 & 0 \\
        \vdots & \vdots & \vdots & \vdots & \ddots & \vdots & \vdots \\
        0 & 0 & 0 & 0 & \cdots & 1 & 1 \\
        0 & 0 & 0 & 0 & \cdots & 0 & 1
    \end{bmatrix}
    = FDF
\]
is the derivative matrix also applied backwards. Here, $F$ is the ``flip'' matrix with $1$'s in entries $(i, n+1-i)$ and $0$ elsewhere. Therefore, 
\begin{align*}
    \op^{-\top}
    &= (FAF) \circ \Pi_{2m}^{-\top} \circ (FDF) \circ \Pi_{2m-1}^{-\top} \circ \dots \circ (FAF) \circ \Pi_2^{-\top} \circ (FDF) \circ \Pi_1^{-\top} \\
    &= F \circ A \circ (F \Pi_{2m}^{-\top} F) \circ D \circ (F \Pi_{2m-1}^{-\top} F) \circ \dots \circ A \circ (F \Pi_2^{-\top} F) \circ D \circ (F\Pi_1^{-\top}) \\
    &= F \circ A \circ \Pi'_{2m} \circ D \circ \Pi'_{2m-1} \circ  \dots \circ A \circ \Pi'_2 \circ D \circ \Pi'_1,
\end{align*}
where $\Pi'_1, \dots, \Pi'_{2m}$ are still random permutations. Since we are interested in the distance of $C_z$, and applying $F$ at the end does not affect the weight of a word, it suffices to understand the distance of the code $C'_z$ generated by applying $\op' := A \circ \Pi'_{2m} \circ D \circ \Pi'_{2m-1} \circ  \dots \circ A \circ \Pi'_2 \circ D \circ \Pi'_1$ to $z_0 \in \widetilde{P}_s$. We call this code a $P_\rep(DA)^m$ code.

Similar codes obtained by serially applying permuted $A$ and $D$ operators to the \emph{repetition code} (instead of the parity code in our case) were proposed and analyzed in \cite{BR26}. They showed that if the outer code is a $2$-repetition code, the resulting $R_2(DA)^m$ and $R_2(AD)^m$ codes are likely (over the choice of random permutations) to have distance approaching the GV bound. 
Our analysis of our $P_\rep(DA)^m$ and $P_\rep(AD)^m$ codes will follow the same line of attack laid out by \cite{BR26}. However, we will not be able to black-box their analysis for a couple reasons: First, they wrote their analysis specifically for the case of the outer encoder being the $2$-repetition code, whereas we will need to generalize their analysis to our $\rep$-parity codes. Second, we found several bugs within their proofs (fortunately, none fatal), and will provide corrected proofs for our setting. 

In the rest of this paper, we will analyze the distance of $P_\rep(DA)^m$ and $P_\rep(AD)^m$ codes. In particular, we will show:

\begin{theorem}
    Let $\rep \ge 2$ be an even integer. Let $\delta_{GV} \in [0, \frac12]$ be such that $h(\delta_{GV}) = 1/\rep$. There is an increasing sequence $\{ \delta^{(m)} \}_{m \ge 2}$ of positive values converging to $\delta_{GV}$, such that for any $\eps > 0$, 
    \begin{itemize}
        \item For $m \ge 2$, the resulting $P_\rep(DA)^m$ code has rate $1 - \frac1\rep$ and relative distance $\ge \delta^{(m)} - \eps$ with probability $\ge 1 - \widetilde{O}(n^{1-m})$ over the choice of $\Pi_1, \dots, \Pi_{2m}$. 
        \item For $m \ge 3$, the resulting $P_\rep(AD)^m$ code has rate $1 - \frac1\rep$ and relative distance $\ge \delta^{(m)} - \eps$ with probability $\ge 1 - \widetilde{O}(n^{2-m})$ over the choice of $\Pi_1, \dots, \Pi_{2m}$.
    \end{itemize}
\end{theorem}

\begin{corollary}[$\calC$ Approaches the CSS GV Bound] \label{cor:css-gv}
    For any even $\rep \ge 4$ and $\eps > 0$, there is some $m \ge 3$ such that the quantum code $\calC$ has rate $1 - 2/\rep$ and also has both $X$-distance and $Z$-distance $\ge \left( h^{-1}(1/\rep) - \eps \right) \cdot n$ with probability $\ge 1 - \widetilde{O}(n^{2-m})$ over the choice of uniformly random $\Pi_1, \dots, \Pi_{2m}$. 
\end{corollary}

\section{Distance Analysis: Setup and Definitions}

\subsection{Average Weight Enumerators} \label{sec:enumerators}

The analysis proceeds by computing the \emph{average weight enumerators} $W_{AD}^{(m)}, W_{DA}^{(m)} : \{ 0, \dots, n \} \rightarrow \bbR$, where $W_{AD}^{(m)}(i)$ is the expected number of weight-$i$ codewords in a $P_\rep(AD)^m$ code and $W_{DA}^{(m)}(i)$ is the expected number of weight-$i$ codewords in a $P_\rep(DA)^m$ code. Here the randomness is over choice of random permutations $\Pi_1, \dots, \Pi_{2m}$. 

Our goal is to bound $W_{AD}^{(m)}(i)$ (and similar for $W_{DA}$). If we can show that $\sum_{0 < i < \delta} W_{AD}^{(m)}(i) \ll 1$ for some $\delta > 0$, then by Markov's inequality it follows that with high probability a randomly generated code will have distance $\ge \delta$. 

The base case of $m = 0$ is simply the weight enumerator for the code $\widetilde{P}_\rep$, which we denote by $\Phi : \{ 0, \dots, n \} \rightarrow \mathbb{Z}_{\ge 0}$. We note that $\widetilde{P}_\rep$ has no odd weight codewords, so $\Phi(c) = 0$ for all odd $c$. 

Let 
\[
    p_A(a, b) = \frac{\binom{n-b}{\lfloor a/2 \rfloor} \binom{b-1}{\lceil a/2 \rceil - 1}}{\binom{n}{a}}
    = \frac{\binom{n-a}{b - \lceil a/2 \rceil} \binom{a}{\lceil a/2 \rceil}}{\binom{n}{b}} \cdot \frac{\lceil a/2 \rceil}{b} \numberthis \label{eqn:pA}
\]
denote the expected number of weight-$b$ codewords mapped to by a given word of weight $a$ after applying a random permutation $\Pi$ followed by the accumulation operator $A$. Note that because of the random permutation, this probability does not depend on the specific starting word and only on the weights $a$ and $b$. Similarly, define
\[
    p_D(a, b) = \frac{\binom{n-a}{\lfloor b/2 \rfloor} \binom{a-1}{\lceil b/2 \rceil - 1}}{\binom{n}{a}}
    = \frac{\binom{n-a}{\lfloor b/2 \rfloor} \binom{a}{\lceil b/2 \rceil}}{\binom{n}{a}} \frac{\lceil b/2 \rceil}{a} \numberthis \label{eqn:pD}
\]
to be the expected number of weight-$b$ words obtained from a word of weight $a$ by applying $\Pi$ followed by $D$. 

We can write recursions for $W_{AD}^{(m)}$ and $W_{DA}^{(m)}$:
\begin{align*}
    W_{AD}^{(m)}(c) 
    &= \sum_{0 \le a \le n} \sum_{0 \le b \le n} W_{AD}^{(m-1)}(a) \cdot p_A(a, b) \cdot p_D(b, c) \numberthis \label{eqn:recursion-WAD} \\
    W_{DA}^{(m)}(c) 
    &= \sum_{0 \le a \le n} \sum_{0 \le b \le n} W_{DA}^{(m-1)}(a) \cdot p_D(a, b) \cdot p_A(b, c). \numberthis \label{eqn:recursion-WDA}
\end{align*}
We remark that for $c > 0$, we need only sum over weights $a, b > 0$. This is because the only $0$-weight codeword is $0^n$, which will only ever map to itself under $A$ and $D$ operations. Since we will be interested in bounding the number of \emph{positive} weight codewords under $\delta$, we can generally exclude $0$ from the range for the recursion.


\subsection{Spectral Shape Analysis}


We will analyze the \emph{spectral shape} of the functions $W_{AD}^{(m)}$ and $W_{DA}^{(m)}$. That is, we will study functions $r_{DA}^{(m)}, r_{AD}^{(m)}: [0, 1] \rightarrow \mathbb{R}$ such that
\begin{align*}
    W_{AD}^{(m)}(c) &\le \poly(n) \cdot 2^{n \cdot r_{AD}^{(m)}(c/n)} \\
    W_{DA}^{(m)}(c) &\le \poly(n) \cdot 2^{n \cdot r_{DA}^{(m)}(c/n)}.
\end{align*}
For $m = 0$, we have that $W_{AD}^{(0)}(c) = W_{DA}^{(0)}(c) = \Phi(c)$, where we defined $\Phi$ to be the weight enumerator of the code $\widetilde{P}_\rep$. Ideally, we could set $r_{AD}^{(0)}(\gamma) = r_{DA}^{(0)}(\gamma)$ to be exactly $\frac1n$ of the logarithm of $\Phi(\gamma n)$, however, we do not know a nice closed form for this expression. Instead, we will choose a slightly loose but convenient upper bound on $\log_2 \Phi(\gamma n)$, defined below.

\begin{definition}[$\varphi$] \label{def:varphi}
    We define
    \[
        \varphi(\gamma) := \min \left\{ h(\gamma) + \frac1\rep \cdot \log_2 \left( \frac{1 + (1 - 2\gamma)^\rep}{2} \right),~~ \frac{\log_2(\rep-1)}{2} \cdot \min \{ \gamma, 1 - \gamma \} + \frac12 \cdot h(\gamma) \right\},
    \]
    where recall $R = 1 - 1/\rep$ is the rate of $\widetilde{P}_\rep$. 
\end{definition}

The first term in the minimization describes the general behavior of $\log_2 \Phi(\gamma n)$ when $\gamma$ is bounded away from $0$ or $1$. The latter term is a tighter bound when $\gamma$ is very close to $0$ or $1$, and will be necessary to show convergence to the GV bound (in particular, for showing that the candidate distances $\delta^{(m)}$ defined below become strictly positive; see \Cref{prop:delta2-bigger-0}). A corollary of the bounds on $\Phi(c)$ (see \Cref{prop:initial-properties}, \Cref{item:initial-middle-bound-stronger,item:initial-boundary-bound}) is that $\Phi(c) \le \poly(n) \cdot 2^{n \cdot \varphi(c/n)}$. 



Now, let us proceed to the recursion. Asymptotically, using Stirling's approximation, we have loosely that 
\[
    p_A(a, b) \le \poly(n) \cdot 2^{n \cdot f_A(\approx \frac{a}{n}, \approx \frac{b}{n})},
\]
where 
\[
    f_A(\alpha, \beta) = \alpha - h(\beta) + (1-\alpha) \cdot h\Big( \frac{\beta - \alpha/2}{1 - \alpha} \Big),
\]
and 
\[
    p_D(a, b) \le \poly(n) \cdot 2^{n \cdot f_D(\approx \frac{a}{n}, \approx \frac{b}{n})}
\]
where 
\[
    f_D(\alpha, \beta) = (1 - \alpha) \cdot h \Big( \frac{\beta}{2(1-\alpha)} \Big) + \alpha \cdot h \Big( \frac{\beta}{2\alpha} \Big) - h(\alpha).
\]
In the above expressions, we used $\approx a/n$ and $\approx b/n$ in the exponents instead of exact $a/n$ and $b/n$ because of the rounding errors due to the floors and ceilings. For instance, $p_A(1, n) = \frac1n$ but $\alpha = \frac1n, \beta = 1$ falls outside the range for which $f_A(\alpha, \beta)$ is defined. Nonetheless, we will see in \Cref{prop:W-less-than-r} that we can use $f_A + f_D$ and $f_D + f_A$ without error in the arguments as a bound on the spectral shape of the transition probabilities of applying $DA$ or $AD$.

Using $f_A$ and $f_D$, we can write the logarithmic versions of the recursions in~\Cref{eqn:recursion-WAD,eqn:recursion-WDA}:

\begin{definition}
    The \emph{asymptotic spectral shape functions} of a $P_\rep(AD)^m$ and $P_\rep(DA)^m$ code are defined recursively as $r_{AD}^{(0)}(\gamma) = r_{DA}^{(0)}(\gamma) = \varphi(\gamma)$, and for $m \ge 1$,
    \begin{align*}
        r_{AD}^{(m)}(\gamma) 
        &:= \max_{\alpha, \beta} \{ r_{AD}^{(m-1)}(\alpha) + f_{AD}(\alpha, \beta, \gamma) \} \\
        r_{DA}^{(m)}(\gamma) 
        &:= \max_{\alpha, \beta} \{ r_{DA}^{(m-1)}(\alpha) + f_{DA}(\alpha, \beta, \gamma) \}, 
    \end{align*}
    Here, we've used the shorthands $f_{AD}(\alpha, \beta, \gamma) = f_A(\alpha, \beta) + f_D(\beta, \gamma)$ and $f_{DA}(\alpha, \beta, \gamma) = f_D(\alpha, \beta) + f_A(\beta, \gamma)$. All intervals for $\alpha, \beta$ should be taken over the valid domains over which $f_A$ and $f_D$ are defined.
\end{definition}

\begin{proposition} \label{prop:W-less-than-r}
    It holds for all $m \ge 0$ and $0 \le c \le n$ that 
    \begin{align*}
        W_{AD}^{(m)}(c) &\le \poly(n) \cdot 2^{n \cdot r_{AD}^{(m)}(c/n)} \\
        W_{DA}^{(m)}(c) &\le \poly(n) \cdot 2^{n \cdot r_{DA}^{(m)}(c/n)}.
    \end{align*}
\end{proposition}

\begin{proof}
    In fact, we will show (and it suffices to show) that 
    \begin{align*}
        \sum_{0 \le b \le n} p_A(a, b) p_D(b, c) &\le \poly(n) \cdot 2^{n \cdot \max_{\beta} f_{AD}(a/n, \beta, c/n)} \numberthis \label{eqn:pApD-fAD-comparison} \\
        \sum_{0 \le b \le n} p_D(a, b) p_A(b, c) &\le \poly(n) \cdot 2^{n \cdot \max_{\beta} f_{DA}(a/n, \beta, c/n)}. \numberthis \label{eqn:pDpA-fDA-comparison}
    \end{align*}
    Then, by induction
    \begin{align*}
        W_{AD}^{(m)}(c) 
        &= \sum_{0 \le a \le n} W_{AD}^{(m-1)}(a) \cdot \sum_{0 \le b \le n} p_A(a,b) p_D(b,c) \\
        &\le \poly(n) \cdot \max_{0 \le a \le n} 2^{n \cdot r_{AD}^{(m-1)}(a/n)} \cdot 2^{n \cdot \max_{\beta} f_{AD}(a/n, \beta, c/n)} \\
        &\le \poly(n) \cdot \max_{\substack{0 \le \alpha \le 1 \\ 0 \le \beta \le 1}} 2^{n \cdot \left( r_{AD}^{(m-1)}(\alpha) + f_{AD}(\alpha, \beta, c/n) \right)} \\
        &= \poly(n) \cdot 2^{n \cdot r_{AD}^{(m)}(c/n)}
    \end{align*}
    as desired. The same holds for $W_{DA}^{(m)}$.

    For \Cref{eqn:pApD-fAD-comparison}, the issue with plugging $\beta = b/n$ into $f_{AD}$ is that $\beta = b/n$ may result in the functions $f_A(a/n, \beta)$ or $f_D(\beta, c/n)$ being undefined (due to the ceilings and floors in $p_A$ and $p_D$). Indeed, $p_A$ restricts $b \in [\frac{a}{2}, n - \frac{a}{2} + \frac12]$ and $p_D$ restricts $b \in [\frac{c}{2}, n - \frac{c}{2} + \frac12]$, but the functions $f_A$ and $f_D$ have ranges $\beta \in [\frac{a}{2n}, 1 - \frac{a}{2n}]$ and $\beta \in [\frac{c}{2n}, 1 - \frac{c}{2n}]$ respectively. Notice that these intervals are nearly identical (after normalization by $n$), up to a $\pm\frac{1}{2n}$ in the normalized upper bound. Therefore, when $b/n$ already falls in the range for $f_A$ and $f_D$, Stirling's approximation immediately gives $p_A(a,b) p_D(b,c) \le \poly(n) \cdot 2^{n \cdot f_{AD}(a/n, b/n, c/n)}$. When $b/n$ does not, defining $\beta = \frac{b}{n} - \frac{1}{n}$ places $\beta$ within the defined range for $f_A$ and $f_D$, and we get that $p_A(a, b) p_D(b, c) \le \poly(n) \cdot p_A(a, b-1) p_D(b-1, c) \le \poly(n) \cdot 2^{n \cdot f_{AD}(a/n, \beta, c/n)}$.

    We will handle \Cref{eqn:pDpA-fDA-comparison} similarly. Here again, the issue is that plugging in $\beta = b/n$ may result in the functions $f_D(a/n, \beta)$ or $f_A(\beta, c/n)$ being undefined. The functions $p_D(a, b)$ and $p_A(b, c)$ restrict $b \le \min \{ 2a, 2(n-a)+1, 2c, 2(n-c)+1 \}$, whereas $f_D(a/n, \beta)$ and $f_A(\beta, c/n)$ only restrict $\beta \le \min \{ \frac{2a}{n}, \frac{2(n-a)}{n}, \frac{2c}{n}, \frac{2(n-c)}{n} \}$. So, when $b/n$ falls within the valid range for $f_D$ and $f_A$, Stirling's approximation gives the desired conclusion. Otherwise, we must round $\beta$ from $b/n$ by at most $\frac{1}{n}$, which will affect the approximation by at most a polynomial factor.
\end{proof}

\cite{BR26} analyzes $r_{AD}^{(m)}$ and $r_{DA}^{(m)}$ by first relating these quantities to the analogous quantity for $R_2(A)^m$ codes where only permuted accumulations (no permuted derivatives) are applied, which have been intensively studied. We will follow their approach and relate our quantities to the code $P_\rep(A)^m$ which is encoded by taking words in $\widetilde{P}_\rep$ and applying $m$ rounds of a random permutation $\Pi$ followed by an accumulation operation $A$. That is, $P_\rep(A)^m = A \circ \Pi_m \circ \dots \circ A \circ \Pi_1 \circ \widetilde{P}_\rep$. The analogous quantity is $r_A^{(m)}$, defined as follows:

\begin{definition}
    The \emph{asymptotic spectral shape function} of a $P_\rep(A)^m$ code is defined recursively as
    \[
        r_A^{(m)}(\gamma) := \max_{0 \le \alpha \le \min \{ 2\gamma, 2(1-\gamma) \} } \{ r_A^{(m-1)}(\alpha) + f_A(\alpha, \gamma) \}
    \]
    with base case $r_A^{(0)}(\gamma) = \varphi(\gamma)$. 
    We also define 
    \[
        \delta^{(m)} := \max \{ \delta \in [0, 1/2) : \forall \delta' \le \delta, r_A^{(m)}(\delta') \le 0 \}.
    \]
\end{definition}

Intuitively, $\delta^{(m)}$ represents the threshold below which there are sub-exponentially many codewords in $P_\rep(A)^m$. It turns out that $\delta^{(m)}$ will also represent the distance of $P_\rep(A)^m$ codes. (Note that we have not yet shown that below $\delta^{(m)}$ there should be no codewords; a priori there may be as many as sub-exponentially many codewords. Bounding the expected number of codewords below $\delta^{(m)}$ will be an important part of our analysis.) We will show that the distance of $P_\rep(DA)^m$ and $P_\rep(AD)^m$ codes approaches $\delta^{(m)}$ as well.

\subsection{Properties of the Initial Encoder, Spectral Shape, and Update Function} \label{sec:properties}

In this section, we state properties about $\Phi$, $f_A$, and $r_A^{(m)}$ that we will use later. Their proofs are given in~\Cref{sec:appendix-properties}.

\begin{restatable}[Facts about $\Phi$ and $\varphi$]{proposition}{phifacts}
\label{prop:initial-properties}
    The following properties hold for the outer encoder. Recall that we defined $\varphi(\gamma) := \min \left\{ h(\gamma) + \frac1\rep \cdot \log_2 \left( \frac{1 + (1 - 2\gamma)^\rep}{2} \right),~ \frac{\log_2(\rep-1)}{2} \cdot \min \{ \gamma, 1 - \gamma \} + \frac12 \cdot h(\gamma) \right\}$.
    \begin{enumerate}
    \item \label{item:initial-symmetric}
        $\Phi(a)$ is symmetric around $a = n/2$. 
    \item \label{item:initial-0-at-1}
        $\Phi(1) = \Phi(n-1) = 0$.
    \item \label{item:initial-boundary-bound}
        $\Phi(a) \le (\rep-1)^{a'/2} \binom{n/2}{a/2}$, where $a' := \min \{ a, n - a \}$. This is the tighter bound when $a$ is close to $0$ or $n$.
    \item \label{item:initial-middle-bound-stronger}
        $\Phi(a) \cdot \alpha^a (1 - \alpha)^{n-a} \le \left( \frac{1 + (1 - 2\alpha)^\rep}{2} \right)^{n/\rep}$. This is the tighter bound when $a$ is bounded away from $0$ and $n$. We will only need this bound to define $\varphi$ so that \Cref{item:initial-phi-h-decreasing} holds. For many of the bounds in this paper, we will only need the below weaker bound (\Cref{item:initial-middle-bound}) on $\Phi(a)$. 
    \item \label{item:initial-middle-bound}
        $\Phi(a) \le \poly(n) \cdot \binom{n}{a}^{1 - 1/\rep}$.
    \item \label{item:initial-phi-increasing}
        $\varphi(\gamma)$ is symmetric around $\frac12$, and increasing between $0$ and $\frac12$.
    \item \label{item:initial-phi-h-decreasing}
        $\varphi(\gamma) - h(\gamma)$ is symmetric around $\frac12$, and decreasing between $0$ and $\frac12$. 
    \item \label{item:initial->h-1/r}
        $\varphi(\gamma) \ge h(\gamma) - 1/\rep$.
    \end{enumerate}
\end{restatable}

We remark that one could pick a different outer encoder $\Phi$, and the analysis should mostly work assuming one proves the statements (or similar statements to the ones) above. Let us explain the role of each statement in the proofs. \Cref{item:initial-symmetric,item:initial-phi-increasing,item:initial-phi-h-decreasing} are needed to prove the necessary facts about $r_A^{(m)}$ in \Cref{prop:rA-facts}. \Cref{item:initial-0-at-1,item:initial-boundary-bound,item:initial-middle-bound} are used to get bounds on the number of codewords with weight super close to $0$ or $n$ after $1$ or $2$ rounds (see \Cref{sec:boundary}). As mentioned above, \Cref{item:initial-middle-bound-stronger} is a tighter bound on the number of middle weight codewords that is useful to prove~\Cref{item:initial-phi-h-decreasing}. Finally, \Cref{item:initial->h-1/r} is used in~\Cref{sec:convergence-to-GV} as a lower bound on $\varphi$ to show that $\delta^{(m)}$ go to the desired GV bound as $m$ grows. Also necessary to the GV bound convergence is \Cref{item:initial-boundary-bound}, which shows that there are quite few codewords with weight close to $0$ or $n$, and is used to show that $r_A^{(2)}$ has a positive length $0$ stretch starting at $0$ (i.e., $\delta^{(2)}$ is strictly positive). 

The below facts about $f_A$ were mostly proved in~\cite{BR26}.

\begin{restatable}[Facts about $f_A$]{proposition}{fAfacts} \label{prop:fA-facts}
    The following hold for $f_A$. 
    \begin{enumerate}
    \item \label{item:fA-zeros}
        $f_A(\alpha, \gamma) \le 0$, with equality if and only if $\alpha = 0$ or $\gamma = \frac12$. 
    \item 
        $f_A(\alpha, \gamma)$ is decreasing in $\alpha$. 
    \item \label{item:fA-increasing-in-gamma}
        $f_A(\alpha, \gamma)$ is symmetric around $\gamma = \frac12$ and is increasing in $\gamma$ between $0$ and $\frac12$. 
    \item \label{item:h-stable}
        For fixed $\gamma \in (0, 1)$, the function $h(\alpha) + f_A(\alpha, \gamma) - h(\gamma) : (0, \min \{ 2\gamma, 2(1-\gamma) \}] \rightarrow \mathbb{R}$ is concave and is maximized at $\alpha = \lambda_\gamma := 2\gamma(1-\gamma)$, at which point $h(\lambda_\gamma) + f_A(\lambda_\gamma, \gamma) - h(\gamma) = 0$. 
    \end{enumerate}
\end{restatable}

Finally, the below facts about $r_A^{(m)}$ follow from the properties stated above about $\varphi$ and about $f_A$.

\begin{restatable}[Facts about $r_A^{(m)}$]{proposition}{rAfacts} \label{prop:rA-facts}
    The following hold for $r_A^{(m)}$ for $m \ge 0$.
    \begin{enumerate}
    \item \label{item:rm-positive}
        $r_A^{(m)}(\gamma) \ge 0$.
    \item \label{item:rm-increasing}
        $r_A^{(m)}(\gamma)$ is symmetric around $\frac12$, and is increasing between $0$ and $\frac12$.
    \item \label{item:rm-h-decreasing}
        $r_A^{(m)}(\gamma) - h(\gamma)$ is symmetric around $\frac12$ and is decreasing between $0$ and $\frac12$.
    \item \label{item:rm-decreasing-in-m}
        For $m \ge 1$, $r_A^{(m)}(\gamma) \le r_A^{(m-1)}(\gamma)$. In particular, $\delta^{(m)} \ge \delta^{(m-1)}$ for all $m \ge 1$. 
    \end{enumerate}
\end{restatable}



\section{Convergence of $\delta^{(m)}$ to the GV Bound} \label{sec:convergence-to-GV}

In this section, we will show that the sequence $\{ \delta^{(m)} \}_{m \ge 0}$ converges to $\delta_{GV}$, defined as the value $\in [0, \frac12]$ for which $h(\delta_{GV}) = 1 - R = 1/\rep$. 

\begin{theorem}
    $\{ \delta^{(m)} \}_{m \ge 0} \rightarrow \delta_{GV}$.
\end{theorem}

\noindent
This theorem will be proved by a combination of \Cref{prop:delta-increasing,prop:delta-less-than-GV,prop:delta-approach-GV}. 

\begin{proposition} \label{prop:delta-increasing}
    $\delta^{(m)} \ge \delta^{(m-1)}$ for all $m > 0$.
\end{proposition}

\begin{proof}
    This follows from $r_A^{(m)} \le r_A^{(m-1)}$ (\Cref{prop:rA-facts}, \Cref{item:rm-decreasing-in-m}). In particular, if $r_A^{(m-1)}(\gamma) = 0$ for all $\gamma \le \delta^{(m-1)}$, then the same will be true for $r_A^{(m)}$. 
\end{proof}

\begin{proposition} \label{prop:delta-less-than-GV}
    $\delta^{(m)} \le \delta_{GV}$ for all $m \ge 0$.
\end{proposition}

\begin{proof}
    We will show via induction that $r_A^{(m)}(\gamma) \ge h(\gamma) - 1/\rep$. The case of $m = 0$ is \Cref{prop:initial-properties}, \Cref{item:initial->h-1/r}. For $m \ge 1$, suppose that $r_A^{(m-1)}(\gamma) \ge h(\gamma) - 1/\rep$. Then
    \begin{align*}
        r_A^{(m)}(\gamma) 
        &= \max_{0 \le \alpha \le \min \{ 2\gamma, 2(1-\gamma) \}} \{ r_A^{(m-1)}(\alpha) + f_A(\alpha, \gamma) \} \\
        &\ge \max_{0 \le \alpha \le \min \{ 2\gamma, 2(1-\gamma) \}} \{ h(\alpha) - 1/\rep + f_A(\alpha, \gamma) \} \\
        &= h(\gamma) - 1/\rep,
    \end{align*}
    where in the last line we use \Cref{prop:fA-facts}, \Cref{item:h-stable}.

    Recall that $\delta_{GV} \in [0, \frac12]$ is defined such that $h(\delta_{GV}) = 1 - R = 1/\rep$. Then for any $\gamma \in (\delta_{GV}, \frac12]$, we have that $r_A^{(m)}(\gamma) > 0$, implying that $\delta^{(m)} \le \delta_{GV}$.  
\end{proof}

\begin{proposition} \label{prop:delta2-bigger-0}
    $\delta^{(2)} > 0$.
\end{proposition}

\begin{proof}
    Let us first bound $r_A^{(1)}(\gamma)$ for $\gamma \le \frac12$ (recall $\gamma > \frac12$ is handled by $r_A^{(1)}$ being symmetric around $\frac12$). We have that 
    \begin{align*}
        r_A^{(1)}(\gamma) 
        &= \max_{0 \le \alpha \le 2\gamma} \{ \varphi(\alpha) + f_A(\alpha, \gamma) \} \\
        &\le \max_{0 \le \alpha \le 2\gamma} \left\{ \frac{\alpha}{2} \cdot \log_2(\rep - 1) + \frac12 \cdot h(\alpha) + f_A(\alpha, \gamma) \right\}
    \end{align*}
    Let $F(\alpha) := \frac{\alpha}{2} \cdot \log_2(\rep - 1) + \frac12 \cdot h(\alpha) + f_A(\alpha, \gamma)$ denote the term being maximized. Deriving, we get 
    \begin{align*}
        F'(\alpha) = \frac{\log_2(s-1)}{2} + 1 + \frac12 \log_2 \left(  \frac{(\gamma - \alpha/2)(1 - \gamma - \alpha/2)}{\alpha \cdot (1-\alpha)} \right).
    \end{align*}
    Because $\alpha \le 2\gamma$, we have that $\frac{1 - \gamma - \alpha/2}{1 - \alpha} \le 1$, so the above is 
    \begin{align*}
        F'(\alpha) 
        &\le \frac{\log_2(s-1) + 1}{2} + \frac12 \log_2 \left( \frac{\gamma - \alpha/2}{\alpha/2} \right) \\
        &= \frac{\partial}{\partial \alpha} \left( \left( \frac{\log_2(s-1) + 1}{2} \right) \cdot \alpha + \gamma \cdot h\left( \frac{\alpha/2}{\gamma} \right) \right) \\
        &=: \frac{\partial}{\partial \alpha} G(\alpha, \gamma).
    \end{align*}
    Since $G(0, \gamma) = 0 = F(0)$, we have that $F(\alpha) \le G(\alpha, \gamma)$. Furthermore, $G$ is concave in $\alpha$, and its derivative is $0$ when $\alpha = \frac{4(\rep-1)}{2(\rep-1) + 1} \cdot \gamma$. Plugging this maximum point into $G(\alpha, \gamma)$, we get $G(\alpha, \gamma) \le \gamma \cdot \log_2(2\rep-1)$. 
    
    Next, Equation 14 in~\cite{FagnaniR08} is the bound
    \[
        f_A(\alpha, \gamma) \le \alpha \log_2 \left( 2 \sqrt{\gamma(1 - \gamma)} \right),
    \]
    from which we obtain
    \begin{align*}
        r_A^{(2)}(\gamma) 
        &= \max_{0 \le \alpha \le 2\gamma} \{ r_A^{(1)}(\alpha) + f_A(\alpha, \gamma) \} \\
        &\le \max_{0 \le \alpha \le 2\gamma} \left\{ \alpha \log_2 (2\rep - 1) + \alpha \log_2 \left( 2 \sqrt{\gamma(1 - \gamma)} \right) \right\} \\
        &= \max_{0 \le \alpha \le 2\gamma} \left\{ \alpha \log_2 \left( 2(2\rep-1) \sqrt{\gamma(1 - \gamma)} \right) \right\}. \numberthis \label{eqn:rA2-bound}
    \end{align*}
    Therefore, whenever $\sqrt{\gamma(1 - \gamma)} < \frac{1}{2(2\rep-1)}$, the quantity being maximized in \Cref{eqn:rA2-bound} is $\le 0$ for all $\alpha$, giving that $\delta^{(2)} \ge \frac{1 - \sqrt{1 - 1/(2\rep-1)^2}}{2} > \frac{1}{4(2\rep-1)^2}$.  
\end{proof}

\begin{proposition} \label{prop:delta-approach-GV}
    For any $\gamma < \delta_{GV}$, there is some $m \ge 0$ for which $r_A^{(m)}(\gamma) = 0$. 
\end{proposition}

\begin{proof}
    Suppose otherwise, and let $\gamma < \delta_{GV}$ be such that $r_A^{(m)}(\gamma) > 0$ for all $m \ge 0$. First, for any $m \ge 1$, we may consider the sequence $\alpha_0, \alpha_1, \dots, \alpha_{m-1}, \alpha_m = \gamma$ such that $r_A^{(m)}(\gamma) > 0$ is obtained recursively by the maxima achieved at each of the $\alpha_i$. That is, $r_A^{(i)}(\alpha_i) = r_A^{(i-1)}(\alpha_{i-1}) + f_A(\alpha_{i-1}, \alpha_i)$ for each $i \in [m]$. Note that $\alpha_i$ is $\le \frac12$ because $r_A^{(i-1)}(\alpha_{i-1})$ is symmetric around $\alpha_{i-1} = \frac12$ and $f_A(\alpha_{i-1}, \alpha_i)$ is decreasing in $\alpha_{i-1}$, so the maximum is always achieved for $\alpha_{i-1} \le \frac12$. Since $f_A(\alpha_{i-1}, \alpha_i) \le 0$, we have that $r_A^{(0)}(\alpha_0) \ge r_A^{(1)}(\alpha_1) \ge \dots \ge r_A^{(m)}(\gamma) > 0$. In particular, $\alpha_i > \delta^{(i)} \ge \delta^{(2)} > 0$ for each $i = 2, \dots, m$. 
    
    Consider any $0 \le i \le m-1$. We have that 
    \begin{align*}
        r_A^{(m)}(\gamma) 
        &= r_A^{(i)}(\alpha_i) + f_A(\alpha_i, \alpha_{i+1}) + \dots + f_A(\alpha_{m-1}, \gamma) \\
        &\le r_A^{(i)}(\alpha_i) + [h(\alpha_{i+1}) - h(\alpha_i)] + \dots + [h(\gamma) - h(\alpha_{m-1})] \\
        &\le r_A^{(i)}(\alpha_i) + h(\gamma) - h(\alpha_i) \\
        &\le h(\gamma) + \frac1\rep \log_2 \left( \frac{1 + (1 - 2\alpha_i)^\rep}{2} \right) \\
        &= - \eps + \frac1\rep \log_2 \left( 1 + (1 - 2\alpha_i)^\rep \right),
    \end{align*}
    where $\eps = \frac1\rep - h(\gamma) > 0$. Here, in the fourth line, we use \Cref{prop:rA-facts}, \Cref{item:rm-decreasing-in-m} to get that $r_A^{(i)}(\alpha_i) \le \varphi(\alpha_i) \le h(\alpha_i) + \frac1\rep \log_2 \left( \frac{1 + (1 - 2\alpha_i)^\rep}{2} \right)$. Using that $r_A^{(m)}(\gamma) > 0$, we get that $\log_2(1 + (1 - 2\alpha_i)^\rep) > \eps\rep$. In particular, there is some $\eps'$ depending only on $\gamma$ (and not $m$) such that $\alpha_i < \frac12 - \eps'$ for all $i = 0, \dots, m-1$. 

    Now, we reach our contradiction. On one hand, since we've assumed that $r_A^{(m)}(\gamma) > 0$ for all $m \ge 0$, we have that for any $m \ge 3$ there is a sequence $\alpha_2, \dots, \alpha_{m-1} \in (\delta^{(2)}, \frac12 - \eps')$ such that 
    \begin{align*}
        r_A^{(m)}(\gamma) 
        &= r_A^{(2)}(\alpha_2) + f_A(\alpha_2, \alpha_3) + \dots + f_A(\alpha_{m-1}, \gamma) \\
        &\le 1 + f_A(\alpha_2, \alpha_3) + \dots + f_A(\alpha_{m-1}, \gamma). 
    \end{align*}
    On the other hand, since $\alpha_i$ are bounded both away from $0$ and from $\frac12$, we have that there is some $\eps'' = \eps''(\delta^{(2)}, \eps') > 0$ for which $f_A(\alpha_{i-1}, \alpha_i) < -\eps''$ (see \Cref{prop:fA-facts}, \Cref{item:fA-zeros}). Then, for $m > 2 + 1/\eps''$, we obtain that $r_A^{(m)}(\gamma) < 0$, a contradiction.
    
\end{proof}

\section{Showing that the Distance Approaches $\delta^{(m)}$}



In this section, we will show that the distance of $P_\rep(DA)^m$ and $P_\rep(AD)^m$ codes in fact approach the prescribed distance $\delta^{(m)}$. To do this, we will show that the expected number of codewords with weight below $(\delta^{(m)}-\eps)n$ is an inverse polynomial function, decaying to $0$ with $n$. 

\begin{theorem}
    Let $\eps > 0$. For $m \ge 2$, the expected number of codewords of a $P_\rep(DA)^m$ code that have weight $\le (\delta^{(m)} - \eps) n$ after $m$ rounds is $\le \tilde{O}(n^{1-m})$. For $m \ge 3$, the expected number of codewords of a $P_\rep(AD)^m$ code that have weight $\le (\delta^{(m)} - \eps) n$ after $m$ rounds is $\le \tilde{O}(n^{2-m})$.

    In particular, for $m \ge 2$, a $P_\rep(DA)^m$ code will have relative distance $> \delta^{(m)} - \eps$ with probability $1 - \widetilde{O}(n^{1-m})$ over the choice of $\Pi_1, \dots, \Pi_{2m}$. Also, for $m \ge 3$, a $P_\rep(AD)^m$ code will have relative distance $> \delta^{(m)} - \eps$ with probability $1 - \widetilde{O}(n^{2-m})$ over the choice of $\Pi_1, \dots, \Pi_{2m}$.
\end{theorem}

How should we go about proving that the expected number of codewords with weight below $(\delta^{(m)} - \eps)n$ is small? Ideally, one could show that $r^{(m)}(\gamma)$ is strictly less than $0$ for all $\gamma < \delta^{(m)}$, which would imply that it's exponentially unlikely for there to be a codeword with weight $\le \delta^{(m)} - \eps$. However, this is false: in fact, $r_A^{(m)}(\gamma) = 0$ for small $\gamma$. Why is this? $r^{(m)}(\gamma) \ge 0$ comes from taking $\alpha = 0$, at which $r_A^{(m-1)}(\alpha) + f_A(\alpha, \gamma) = 0$. When $r^{(m)}(\gamma) = 0$, we only know that the expected number of codewords is at most $\poly(n)$, but we cannot argue the stronger property that the expected number of codewords is $< 1$. However, when $\alpha$ is sufficiently bounded away from $0$, the expression turns negative and one can argue that the contributions from such $\alpha$ are negligible.

Following \cite{FagnaniR08,BR26}, we split the recursion depending on the value of $\alpha$ in the previous round of the recursion. Let $\tau = \log^2(n)/n$, and define 
\begin{align*}
    W_{DA, \ge \tau}^{(m)}(c) &= \sum_{\substack{\tau n \le a \le (1 - \tau)n \\ 0 \le b \le n }} \{ W_{DA}^{(m-1)}(a) \cdot p_D(a,b) \cdot p_A(b,c) \} \\
    W_{DA, < \tau}^{(m)}(c) &= \sum_{\substack{a \in [0, \tau n) \cup ((1-\tau)n, n] \\ 0 \le b \le n }} \{ W_{DA}^{(m-1)}(a) \cdot p_D(a,b) \cdot p_A(b,c) \} \\
    W_{AD, \ge \tau}^{(m)}(c) &= \sum_{\substack{\tau n \le a \le (1 - \tau) n \\ 0 \le b \le n }} \{ W_{AD}^{(m-1)}(a) \cdot p_A(a,b) \cdot p_D(b,c) \} \\
    W_{AD, < \tau}^{(m)}(c) &= \sum_{\substack{a \in [0, \tau n) \cup ((1-\tau) n, n] \\ 0 \le b \le n }} \{ W_{AD}^{(m-1)}(a) \cdot p_A(a,b) \cdot p_D(b,c) \} 
\end{align*}
Then, we can split each weight enumerator as
\begin{align*}
    W_{DA}^{(m)}(c) &= W_{DA, \ge \tau}^{(m)}(c) + W_{DA, < \tau}^{(m)}(c) \\
    W_{AD}^{(m)}(c) &= W_{AD, \ge \tau}^{(m)}(c) + W_{AD, < \tau}^{(m)}(c)
\end{align*}
and bound the two terms individually. We call the terms $W_{\ge \tau}^{(m)}$ the \emph{middle contribution}, and the $W_{< \tau}^{(m)}$ terms the \emph{boundary contribution}.

\subsection{Bounding the Middle Contribution} \label{sec:middle}

\begin{theorem} \label{thm:middle}
    Let $\eps > 0$ and $m \ge 1$, and let $\tau = \log^2(n) / n$. The expected number of nonzero codewords of a $P_\rep(AD)^m$ or $P_\rep(DA)^m$ code that have weight between $\tau n$ and $(1 - \tau)n$ after $m-1$ rounds and weight $\le (\delta^{(m)} - \eps)n$ or $\ge (1 - \delta^{(m)} + \eps) n$ after $m$ rounds is $\le \negl(n)$. 
\end{theorem}

The rest of the section will be dedicated to proving this theorem. We will prove the statement for $P_\rep(DA)^m$ codes; the case of $P_\rep(AD)^m$ codes is identical. 

Let's define the restricted asymptotic shape functions, depending on the value of $\alpha$ after the $(m-1)$'th round of encoding. 
\begin{align*}
    r_{A, \ge \tau}^{(m)}(\gamma) &= \max_{\tau \le \alpha \le 1 - \tau} \{ r_A^{(m-1)}(\alpha) + f_{A}(\alpha, \gamma) \} \\
    r_{DA, \ge \tau}^{(m)}(\gamma) &= \max_{\substack{\tau \le \alpha \le 1 - \tau \\ 0 \le \beta \le 1 }} \{ r_{DA}^{(m-1)}(\alpha) + f_{DA}(\alpha, \beta, \gamma) \} \\
    r_{AD, \ge \tau}^{(m)}(\gamma) &= \max_{\substack{\tau \le \alpha \le 1 - \tau \\ 0 \le \beta \le 1 }} \{ r_{AD}^{(m-1)}(\alpha) + f_{AD}(\alpha, \beta, \gamma) \}.
\end{align*}
Each of the above expressions should be read to be maximized over the intersection of the listed ranges and the domains over which $f_A$ and $f_D$ are defined. 

We'd like to show that $r_{DA, \ge \tau}^{(m)}(\gamma) = -\Omega(\log(n)^2/n)$ for each $\gamma \le \delta^{(m)} - \eps$ or $\gamma \ge 1 - \delta^{(m)} + \eps$. If we could show this, then the expected number of codewords with weight in $[\tau n, (1-\tau) n]$ after $m-1$ rounds and weight $\le (\delta^{(m)} - \eps)n$ or $\ge (1 - \delta^{(m)} + \eps) n$ after $m$ rounds is
\begin{align*}
    \sum_{\substack{c \in [1, (\delta^{(m)} - \eps) n] \\ \cup [(1 - \delta^{(m)} + \eps) n, n]}} W_{DA, \ge \tau}^{(m)}(c) 
    &\le \poly(n) \cdot \sum_{\substack{c \in [1, (\delta^{(m)} - \eps) n] \\ \cup [(1 - \delta^{(m)} + \eps) n, n]}} 2^{n \cdot r_{DA, \ge \tau}^{(m)}(c/n)} \\
    &\le \poly(n) \cdot 2^{n \cdot (-\Omega(\log^2(n)/n))} \\
    &\le \negl(n),
\end{align*}
as desired.

Letting $\lambda_\gamma := 2\gamma(1-\gamma)$, define
\[
    g(\alpha, \gamma) = \begin{cases}
        f_A(\alpha, \gamma) & \alpha \le \lambda_\gamma \\
        h(\gamma) - h(\alpha) & \lambda_\gamma \le \alpha \le 1 - \lambda_\gamma \\
        f_A(1-\alpha, \gamma) & \alpha \ge 1 - \lambda_\gamma.
    \end{cases}
\]

\begin{restatable}[\cite{BR26}, Theorems 4.3 and 4.5]{lemma}{fvsg}  \label{lem:fDA<g}
    \[
        f_{DA}(\alpha, \beta, \gamma) \le g(\alpha, \gamma), 
        \qquad f_{AD}(\alpha, \beta, \gamma) \le g(\alpha, \gamma).
    \]
\end{restatable}

Using these bounds, we're able to obtain a comparison of $r_{DA, \ge \tau}(\gamma)$ to $r_{A, \ge \min \{ \tau, \lambda_\gamma \}}(\gamma)$. A version of the following lemma was given in Theorem 4.7 of \cite{BR26}, but there appears to be an error in their analysis. We provide a corrected statement below. The missing assumption is that $r_A^{(m)}(\gamma) - h(\gamma)$ is symmetric around $\frac12$ and non-increasing on its support in $[0, \frac12]$: this is \Cref{prop:initial-properties}, \Cref{item:initial-phi-h-decreasing} and \Cref{prop:rA-facts}, \Cref{item:rm-h-decreasing}.

\begin{restatable}{lemma}{rDAvsrA} \label{lem:rDAvsrA}
    If $\varphi(x) - h(x)$ is symmetric around $\frac12$ and non-increasing on its support in $[0, \frac12]$, then 
    \[
        r_{DA, \ge \tau}^{(m)}(\gamma) 
        \le \max_{\tau \le \alpha \le 1 - \tau} \{ r_{DA}^{(m-1)}(\alpha) + g(\alpha, \gamma) \} 
        \le r_{A, \ge \min \{ \tau, \lambda_\gamma \}}^{(m)}(\gamma)  
    \]
    and
    \[
        r_{AD, \ge \tau}^{(m)}(\gamma) 
        \le \max_{\tau \le \alpha \le 1 - \tau} \{ r_{AD}^{(m-1)}(\alpha) + g(\alpha, \gamma) \}
        \le r_{A, \ge \min \{ \tau, \lambda_\gamma \}}^{(m)}(\gamma).
    \]
\end{restatable}

\begin{proof}
    Let us handle the case of $DA$; $AD$ will follow the exact same proof. We prove this by induction on $m$. For $m = 0$, we have that 
    \[
        r_{DA}^{(0)}(\gamma) = \varphi(\gamma) = r_A^{(0)}(\gamma),
    \]
    establishing the base case.

    For $m > 0$, assume that we've already established 
    \[
        r_{DA, \ge \tau}^{(m-1)}(\gamma) \le r_{A, \ge \min \{ \tau, \lambda_\gamma \}}^{(m-1)}(\gamma).
    \]
    Then, 
    \begin{align*}
        r_{DA, \ge \tau}^{(m)}(\gamma) 
        &\le \max_{\substack{\tau \le \alpha \le 1 - \tau \\ 0 \le \beta \le 1 }} \left\{ r_{DA}^{(m-1)}(\alpha) + f_{DA}(\alpha, \beta, \gamma) \right\} \\
        &\le \max_{\tau \le \alpha \le 1 - \tau} \left\{ r_{DA}^{(m-1)}(\alpha) + g(\alpha, \gamma) \right\} \\
        &\le \max_{\tau \le \alpha \le 1 - \tau} \left\{ r_A^{(m-1)}(\alpha) + g(\alpha, \gamma) \right\}, \numberthis \label{eqn:rA+g}
    \end{align*}
    where we used in the second inequality that $f_{DA}(\alpha, \beta, \gamma) \le g(\alpha, \gamma)$ by \Cref{lem:fDA<g}, and in the last inequality we used the inductive hypothesis with $\tau = 0$ to get that $r_{DA}^{(m-1)}(\alpha) \le r_A^{(m-1)}(\alpha)$. 
    We split \Cref{eqn:rA+g} into three intervals and show each is $\le r_{A, \ge \min \{ \tau, \lambda_\gamma \}}^{(m)}(\gamma)$, giving the desired conclusion.
    \begin{itemize}
    \item 
        The first interval is $\min(\tau, \lambda_\gamma) \le \alpha \le \lambda_\gamma$. In this interval, $g(\alpha, \gamma) = f_A(\alpha, \gamma)$, so 
        \begin{align*}
            \max_{\min(\tau, \lambda_\gamma) \le \alpha \le \lambda_\gamma} \left\{ r_A^{(m-1)}(\alpha) + g(\alpha, \gamma) \right\}
            &= \max_{\min(\tau, \lambda_\gamma) \le \alpha \le \lambda_\gamma} \left\{ r_A^{(m-1)}(\alpha) + f_A(\alpha, \gamma) \right\} \\
            &\le r_{A, \ge \min \{ \tau, \lambda_\gamma \}}^{(m)}(\gamma).
        \end{align*}
    \item 
        The second interval is $\lambda_\gamma \le \alpha \le 1 - \lambda_\gamma$. In this interval, $g(\alpha, \gamma) = h(\gamma) - h(\alpha)$, so
        \begin{align*}
            \max_{\lambda_\gamma \le \alpha \le 1 - \lambda_\gamma} \left\{ r_A^{(m-1)}(\alpha) + g(\alpha, \gamma) \right\}
            &= \max_{\lambda_\gamma \le \alpha \le 1 - \lambda_\gamma} \left\{ r_A^{(m-1)}(\alpha) + h(\gamma) - h(\alpha) \right\}
        \end{align*}
        But \Cref{prop:rA-facts}, \Cref{item:rm-h-decreasing} tells us that the function $r_A^{(m-1)}(\alpha) - h(\alpha)$ is symmetric around $\alpha=\frac12$ and is non-increasing in $[0, \frac12]$, thus is maximized at $\alpha = \lambda_\gamma, 1 - \lambda_\gamma$, giving us that the above is also maximized when $\alpha = \lambda_\gamma, 1 - \lambda_\gamma$. However, recall that $f_A(\lambda_\gamma, \gamma) = h(\gamma) - h(\lambda_\gamma)$, so the above is 
        \[
            = r_A^{(m-1)}(\lambda_\gamma) + f_A(\lambda_\gamma, \gamma) 
            \le r_{A, \ge \min \{ \tau, \lambda_\gamma \}}^{(m)}(\gamma).
        \]
    \item 
        Finally, we have the interval $1 - \lambda_\gamma \le \alpha \le 1 - \min(\tau, \lambda_\gamma)$. Here, we have that $g(\alpha, \gamma) = f_A(1 - \alpha, \gamma)$, so
        \begin{align*}
            \max_{1 - \lambda_\gamma \le \alpha \le 1 - \min(\tau, \lambda_\gamma)} \left\{ r_A^{(m-1)}(\alpha) + g(\alpha, \gamma) \right\}
            &= \max_{1 - \lambda_\gamma \le \alpha \le 1 - \min(\tau, \lambda_\gamma)} \left\{ r_A^{(m-1)}(\alpha) + f_A(1 - \alpha, \gamma) \right\} \\
            &= \max_{\min(\tau, \lambda_\gamma) \le \alpha' \le \lambda_\gamma} \left\{ r_A^{(m-1)}(\alpha') + f_A(\alpha', \gamma) \right\} \\
            &\le r_{A, \ge \min \{ \tau, \lambda_\gamma \}}^{(m)}(\gamma). 
        \end{align*}
        where in the second equality we substituted $\alpha' = 1-\alpha$ and used that $r_A^{(m-1)}(\alpha)$ is symmetric about $\frac12$. 
    \end{itemize}
\end{proof}

It thus suffices to provide bounds on $r_{A, \ge \min \{ \tau, \lambda_\gamma \}}(\gamma)$. Observe that $r_{A, \ge \min \{\tau, \lambda_\gamma \}}(\gamma)$ is symmetric around $\gamma = \frac12$: this follows by inspecting the definition of $r_{A, \ge \min \{ \tau, \lambda_\gamma \} }(\gamma) = \max_{\alpha} \{ r_A^{(m-1)}(\alpha) + f_A(\alpha, \gamma) \}$. Here the range of $\alpha$ is $[\min \{ \tau, \lambda_\gamma \}, \min \{ 2\gamma, 2(1-\gamma), 1 - \min \{ \tau, \lambda_\gamma \} \}]$. It is easy to check that both $f_A(\alpha, \gamma)$ and the range of $\alpha$ are equal for $\gamma$ and $1 - \gamma$. Therefore, we will focus for the rest of this proof only on $\gamma < \frac12$. 

Recall that $\delta^{(m)}$ was defined as the largest $\delta \in [0, \frac12)$ for which $r_A^{(m)}(\gamma) \le 0$ for all $\gamma \le \delta$. Thus $r_{DA}^{(m-1)}(\alpha) + g(\alpha, \gamma) \le 0$ for all $\gamma \le \delta^{(m)}$ and all $\alpha \in [\tau, 1-\tau]$. In particular, let us compare $r_{DA}^{(m-1)}(\alpha) + g(\alpha, \gamma)$ to $r_{DA}^{(m-1)}(\alpha) + g(\alpha, \delta^{(m)}) \le r_{A \ge \min \{ \tau, \lambda_\gamma \}}(\delta^{(m)}) \le 0$. The following will establish that $r_{DA}^{(m-1)}(\alpha) + g(\alpha, \gamma)$ should be bounded strictly below $0$ for all $\alpha$ and $\gamma \le \delta^{(m)} - \eps$ (and therefore also for $\gamma > 1 - \delta^{(m)} + \eps$).

\begin{proposition}
    For fixed $0 < \delta_1 < \delta_2 < \frac12$, there is a constant $c$ independent of $\alpha, n$ such that 
    \[
        g(\alpha, \delta_2) - g(\alpha, \delta_1) \ge c \min \{ \alpha, 1 - \alpha \} \cdot (\delta_2 - \delta_1).
    \]
\end{proposition}

\begin{proof}
    Let $\alpha' = \min \{ \alpha, 1 - \alpha \}$. Then we can write 
    \[
        g(\alpha, \gamma) = \begin{cases}
            f_A(\alpha', \gamma) & \alpha' \le \lambda_\gamma \\
            h(\gamma) - h(\alpha') & \alpha' \ge \lambda_\gamma.
        \end{cases}
    \]
    When $\alpha' \le \lambda_\gamma$ and $\gamma \le \delta_2$, we can calculate
    \[
        \frac{\partial g}{\partial \gamma} (\alpha, \gamma) 
        = \frac{\partial f_A}{\partial \gamma} (\alpha', \gamma)
        = \log_2 \left( 1 + \frac{\alpha' (1 - 2\gamma)}{2(1-\gamma)(\gamma - \alpha'/2)} \right)
        \ge \log_2 \left( 1 + 2\alpha'(1-2\delta_2) \right) 
        \ge \alpha' \cdot \frac{1 - 2\delta_2}{\ln 2},
    \]
    where we used in the first inequality that $(1 - \gamma)(\gamma - \alpha'/2) \le \frac14$ and the second follows from the estimate $\log(1+x) \ge x/2$ for $x < 1$ which follows from looking at the Taylor series for log.

    When $\alpha' \ge \lambda_\gamma$ and $\gamma \le \delta_2$, we can calculate
    \[
        \frac{\partial g}{\partial \gamma} (\alpha, \gamma) 
        = \frac{\partial}{\partial \gamma} [ h(\gamma) - h(\alpha) ]
        = \log_2 \frac{1 - \gamma}{\gamma}
        \ge \log_2 \frac{1 - \delta_2}{\delta_2}.
    \]
    Therefore, taking $c = \min \left\{ \frac{1 - 2\delta_2}{\ln 2}, \log_2 \frac{1 - \delta_2}{\delta_2} \right\} > 0$ gives the desired estimate.
\end{proof} 

Now, to finish the proof, we have that
\begin{align*}
    r_{DA, \ge \tau}^{(m)}(\gamma) 
    &\le \max_{\tau \le \alpha \le 1-\tau} \{ r_{DA}^{(m-1)}(\alpha) + g(\alpha, \gamma) \} \\
    &\le \max_{\tau \le \alpha \le 1-\tau} \{ r_{DA}^{(m-1)}(\alpha) + g(\alpha, \delta^{(m)}) - c \min \{\alpha, 1-\alpha\} \cdot (\delta^{(m)} - \gamma) \} \\
    &\le r_{A, \ge \min\{ \tau, \lambda_\gamma \}}^{(m)}(\delta^{(m)}) + \max_{\tau \le \alpha \le 1-\tau} ( - c \min \{ \alpha, 1 - \alpha \} \cdot (\delta^{(m)} - \gamma) ) \\
    &\le -c \tau \eps \\
    &= -\Omega(\log^2(n)/n).
\end{align*}

\subsection{Bounding the Boundary Contribution} \label{sec:boundary}

In this section, we will bound $\sum_{1 \le c \le (\delta^{(m)} - \eps)n} W_{AD, < \tau}^{(m)} (c)$ and $\sum_{1 \le c \le (\delta^{(m)} - \eps)n} W_{DA, < \tau}^{(m)} (c)$. To do this, we will simply bound the expected number of codewords that have weight $< \tau n$ or $> (1 - \tau)n$ after $m-1$ rounds, with no consideration for what happens in the $m$'th round. As notation, $T := \tau n = \log^2(n)$, and $B_T$ will denote the range $[1, T) \cup (n - T, n]$, so that we are interested in the quantity $\sum_{c \in B_T} W_{DA}^{(m-1)}(c)$ (resp. $\sum_{c \in B_T} W_{AD}^{(m-1)}(c)$). Here, we have not included $0$ in the summation because the only codeword with $0$ weight is the $0$ codeword, which we are not counting. We will also use throughout this section the notation that a variable with a prime indicates the distance of that variable to $0$ or $n$, e.g. $a' = \min \{ a, n - a \}$. 

A similar analysis was done in~\cite{BR26}, but only where the outer encoder is a $2$-repetition code. Furthermore, they attempted to avoid dealing with ceilings and floors by approximating $p_A$ and $p_D$ with functions $p'_A$ and $p'_D$ obtained by removing ceilings and floors and $-1$'s, which they claimed to be a constant factor approximation. However, there appears to be an error in their approximation, and $p_A$ and $p_D$ could be as much as $O(\sqrt{n})$ larger than $p'_A$ and $p'_D$, affecting their final estimates. Below, in \Cref{sec:boundary-DA,sec:boundary-AD}, we repair their proof by retaining the ceilings and floors.

Let us start with some basic bounds on $p_A$ and $p_D$. The proofs of these bounds are in \Cref{sec:appendix-properties}.

\begin{restatable}{proposition}{pAbounds}
\label{prop:bounds}
    In the below statements, $a' := \min \{ a, n - a \}$, $b' := \min \{ b, n - b \}$, and $c' := \min \{ c, n - c \}$. The following bounds hold:
    \begin{enumerate}
    \item \label{item:pA-bound}
        $p_A(a,b) \le \left( \frac{4(b'+1)}{n} \right)^{\lceil a/2 \rceil}$
    \item \label{item:pD-bound2}
        $p_D(a, b) \le \frac{(a'n)^{\lfloor b/2 \rfloor}}{\left( \lfloor b/2 \rfloor! \right)^2 \binom{n}{a}}$
    \item \label{item:pDpA-bound}
        $p_D(a, b) \cdot p_A(b, c) \le \frac{(4a'(c'+1))^{\lfloor b/2 \rfloor}}{(\lfloor b/2 \rfloor!)^2 \binom{n}{a}} \le \frac{(4a'(c'+1))^{\lfloor b/2 \rfloor}}{\binom{n}{a}}$
    \item \label{item:pDpA-trivial-bound}
        $p_A(a, b), p_D(a, b) \le \min \left\{ 1, \frac{\binom{n}{b}}{\binom{n}{a}} \right\}$
    \item \label{item:pT-trivial-bound}
        Let $k \ge 1$ and let $T_1, \dots, T_k$ denote some sequence of $A$'s and $D$'s. It holds that 
        \[
            \Phi(a_0) \prod_{i=1}^k p_{T_{i}}(a_{i-1}, a_i) \le \frac{\poly(n) \cdot \binom{n}{a_k}}{\binom{n}{a_{k-1}}^{1/\rep}}.
        \]
    \end{enumerate}
\end{restatable}

\subsubsection{Boundary Contribution from $P_\rep(DA)^m$ Codes} \label{sec:boundary-DA}

Let us begin with the base case for $m = 1$. 

\begin{lemma} \label{lemma:DA-base-case}
    $\sum_{c \in B_T} W_{DA}^{(1)}(c) = \widetilde{O}(n^{-1})$. 
\end{lemma}

\begin{proof}
    Recall that
    \[
        W_{DA}^{(1)}(c) = \sum_{1 \le a \le n} \sum_{1 \le b \le n} \Phi(a) \cdot p_D(a, b) \cdot p_A(b, c). \numberthis \label{eqn:WDA1}
    \]
    First, let us separate out the summand corresponding to $a = n$. We can compute that the quantity $\sum_{1 \le b \le n} \Phi(n) p_D(n, b) p_A(b, c) = \frac1n$ for all $c$, since after a $D$ operation the weight $n$ string $1^n$ necessarily maps to a weight $b = 1$ string, and $p_A(1, c) = \frac1n$ for all $c \in [n]$. Further, \Cref{prop:initial-properties}, \Cref{item:initial-0-at-1} says that $\Phi(1) = \Phi(n-1) = 0$, so we need only consider $a$ in the range $[2, n-2]$. 
    
    By inspecting the formulas for $p_A(b, c)$ and $p_D(a, b)$ (see \Cref{eqn:pA,eqn:pD}), we observe that $p_A(b, c)$ is only nonzero when $1\le b \le \min \{ 2c, 2(n-c) + 1 \}$ and $p_D(a, b)$ is only nonzero when $1 \le b \le \min \{ 2a, 2(n-a) + 1 \}$. Since we are only interested in $c \in B_T$, we can restrict the range of $b$ to be $[1, \min \{ 2a' + 1, 2 T + 1 \}]$.
    
    Let us bound the quantity $p_D(a, b) p_A(b, c)$. \Cref{prop:bounds}, \Cref{item:pDpA-bound} states that
    \[
        p_D(a, b) \cdot p_A(b, c) 
        \le \frac{(4a'(c'+1))^{\lfloor b/2 \rfloor}}{\binom{n}{a}}. \numberthis \label{eqn:pDpA-bound}
    \]
    Note that $4a'(c'+1) \ge 1$ for $a$ in the range $[1, n-1]$, so the right hand side of \Cref{eqn:pDpA-bound} grows geometrically with $b$. 

    Therefore, using that we need only sum over $1 \le b \le \min \{ 2a' + 1, 2 T + 1 \}$ in \Cref{eqn:WDA1}, we can write 
    \begin{align*}
        W_{DA}^{(1)}(c) 
        &\le \frac1n + \sum_{2 \le a \le n-2} \sum_{1 \le b \le \min \{ 2a' + 1, 2T + 1 \}} \Phi(a) \cdot \frac{(4 a' (c'+1))^{\lfloor b/2 \rfloor}}{\binom{n}{a}} \\
        &\le \frac1n + O(1) \cdot \sum_{2 \le a \le n-2} \Phi(a) \cdot \frac{(4 a' (c' + 1))^{\min \{ a', T \}}}{\binom{n}{a}}.
    \end{align*}
    When $a \in B_T$, we use the bound~\Cref{prop:initial-properties}, \Cref{item:initial-boundary-bound} which, combined with the fact that $\Phi(a)$ is symmetric around $a = n/2$ (\Cref{prop:initial-properties}, \Cref{item:initial-symmetric}), states that $\Phi(a) \le (\rep-1)^{a'/2} \binom{n/2}{a/2} \le \frac{(\rep-1)^{a'/2} \binom{n}{a}}{\binom{n/2}{a/2}}$ for even $a$ (and $\Phi(a) = 0$ when $a$ is odd), and when $a \not\in B_T$, we use the bound~\Cref{prop:initial-properties}, \Cref{item:initial-middle-bound} which states that $\Phi(a) \le \poly(n) \binom{n}{a}^{1 - 1/\rep}$, to get:
    \begin{align*}
        W_{DA}^{(1)}(c) 
        &\le \frac1n + O(1) \cdot \sum_{2 \le a' < T} (\rep-1)^{a'/2} \cdot \frac{(4 a'(T + 1))^{a'}}{\binom{n/2}{a/2}} + \sum_{a \not\in B_T} \poly(n) \cdot \frac{(4 (T + 1) a')^{T}}{\binom{n}{a}^{1/\rep}} \\
        &\le \widetilde{O}(n^{-1}) + n^{-\Omega(T)},
    \end{align*}
    where the final bound comes from observing that the first summation is maximized when $a'$ is small (in particular, $a' = 2$), and that the second summation is dominated by the behavior when $a'$ is as small as possible (close to $T$). 
\end{proof}

Now that we have established the base case for $m = 1$, let us prove via induction that $W_{DA}^{(m)}(c) = \widetilde{O}(n^{-m})$. 

\begin{theorem} \label{thm:DA-boundary}
    For $m \ge 1$, $\sum_{c \in B_T} W_{DA}^{(m)}(c) = \widetilde{O}(n^{-m})$. Therefore, for $m \ge 2$, the expected number of codewords of a $P_\rep(DA)^m$ code that have weight $< \tau n$ or $> (1 - \tau)n$ after $m-1$ rounds is $\widetilde{O}(n^{1-m})$. 
\end{theorem}

\begin{proof}
    The statement for $m = 1$ was given in \Cref{lemma:DA-base-case}. For $m > 1$, we can split
    \[
        W_{DA}^{(m)}(c) 
        \le W_{DA, < \tau}^{(m)}(c) + W_{DA, \ge \tau}^{(m)}(c). 
    \]
    We showed in \Cref{thm:middle} that the second term $W_{DA, \ge \tau}^{(m)}(c)$ is $\le \negl(n)$ for $c' \le (\delta^{(m)} - \eps)n$, which holds in our case since we only care about $c' < T = \log^2(n)$. As for the first term $W_{DA, < \tau}^{(m)}(c) =\sum_{\substack{a \in B_T \\ b \le 2a' + 1}} W_{DA}^{(m-1)}(a) \cdot p_{D}(a, b) p_A(b, c)$, we separate out the case of $a = n$ from the rest of $a \in B_T \backslash \{ n \}$. For $a = n$, the only word is $1^n$, which maps under $D$ to a string of weight $b = 1$, and then maps under $A$ to a weight $c \in [n]$ string with probability $\frac1n$. Therefore, $W_{DA}^{(m-1)}(n) \cdot p_D(n, b) p_A(b, c) = \frac{1}{n} \cdot W_{DA}^{(m-1)}(n)$. For all other $1 \le a' <T$, \Cref{prop:bounds}, \Cref{item:pDpA-bound} gives that 
    \[
        \sum_{b \le 2a' + 1} p_{D}(a, b) p_A(b, c) 
        \le \sum_{b \le 2a' + 1} \frac{(4a'(c'+1))^{\lfloor b/2 \rfloor}}{\binom{n}{a}}
        \le O(1) \cdot \frac{(4a'(c'+1))^{a'}}{\binom{n}{a}} 
        \le \widetilde{O}(n^{-1}).
    \]
    Therefore, the first term can be bounded as
    \begin{align*}
        W_{DA, < \tau}^{(m)}(c) 
        &= \frac1n \cdot W_{DA}^{(m-1)}(n) + \sum_{\substack{a \in B_T \backslash \{ n \} \\ b \le 2a' + 1}} W_{DA}^{(m-1)}(a) \cdot p_{D}(a, b) p_A(b, c) \\
        &\le O(T) \cdot \widetilde{O}(n^{1-m}) \cdot \widetilde{O}(n^{-1}) \\
        &= \widetilde{O}(n^{-m}).
    \end{align*}
\end{proof}

\subsubsection{Boundary Contribution for $P_\rep(AD)^m$ Codes} \label{sec:boundary-AD}

We begin with the base case of $m = 2$. 

\begin{lemma} \label{lemma:AD-base-case}
    $\sum_{e \in B_T} W_{AD}^{(2)}(e) = \widetilde{O}(n^{-1})$
\end{lemma}

\begin{proof}
    Our goal is to bound, for $e \in B_T$, 
    \[
        \sum_{a,b,c,d = 1}^n \Phi(a) p_A(a, b) p_D(b, c) p_A(c, d) p_D(d, e). \numberthis \label{eqn:sumADAD}
    \]
    Let us break this sum into three parts: $d \not\in B_{2\rep T}$; $d \in B_{2\rep T}$ and $b \not\in B_T$; and $d \in B_{2\rep T}$ and $b \in B_T$. 

    \begin{enumerate}

    \item 
        For the case where $d \not\in B_{2\rep T}$, we use the bound \Cref{prop:bounds}, \Cref{item:pT-trivial-bound} which says: 
        \begin{align*}
            \Phi(a) p_A(a, b) p_D(b, c) p_A(c, d) p_D(d, e) 
            &\le \frac{\poly(n) \cdot \binom{n}{e}}{\binom{n}{d}^{1/\rep}}. \numberthis \label{eqn:pApDpApD-bound}
        \end{align*}
        Recall that $e' < T$ and $d' \ge 2\rep T$, so \Cref{eqn:pApDpApD-bound} is $n^{-\Omega(T)}$, and the contribution in \Cref{eqn:sumADAD} for $d \not\in B_{2\rep T}$ is $\poly(n) \cdot n^{-\Omega(T)} = \negl(n)$. 

    \item 

        Next, consider $d \in B_{2\rep T}$ and $b \in B_{T}$. We substitute the following bounds:
        \begin{itemize}
        \item 
            First, using \Cref{prop:initial-properties}, \Cref{item:initial-boundary-bound} along with \Cref{prop:bounds}, \Cref{item:pA-bound}, and noting that $\Phi(a)$ is $0$ for odd $a$, we get
            \begin{align*}
                \sum_{a=1}^n \Phi(a) p_A(a, b) 
                &\le \sum_{a/2 = 1}^{n/2} (\rep-1)^{a/2} \dbinom{n/2}{a/2} \cdot \left( \frac{4(b'+1)}{n} \right)^{a/2} \\
                &= \left( 1 + \frac{4(\rep-1)(b'+1)}{n} \right)^{n/2} - 1 \\
                &\le e^{2(\rep-1)(b'+1)}.
            \end{align*}
            This bound holds for all $b$. When $b = n$, we note that in fact $\sum_{a=1}^n \Phi(a) p_A(a,n) = 0$: the only input $a$ for which $p_A(a,n) \not= 0$ is $a = 1$, but \Cref{prop:initial-properties}, \Cref{item:initial-0-at-1} gives that $\Phi(1) = 0$. Therefore, we can actually restrict the sum to $b \in B_T \backslash \{ n \}$.
        \item 
            Second, \Cref{prop:bounds}, \Cref{item:pDpA-bound} states that $p_D(b, c) \cdot p_A(c, d) \le \frac{(4b'(d'+1))^{\lfloor c/2 \rfloor}}{\binom{n}{b}}$.
            Recalling that $p_D(b,c)$ is only nonzero when $c \le 2b' + 1$, and using that $4b'(d'+1)$ is either $\ge 1$ or $0$ and is hence maximized when $c$ is maximized, we get 
            \begin{align*}
                \sum_{c = 1}^n p_D(b, c) \cdot p_A(c, d) 
                \le O(1) \cdot \frac{(4b'(d'+1))^{b'}}{\binom{n}{b}}.
            \end{align*}
        \item 
            Finally, $p_D(d, e) \le 1$.
        \end{itemize}
        Combining the three above bounds, we obtain 
        \begin{align*}
            \sum_{a \in [n]} \sum_{b \in B_T} & \sum_{c \in [n]} \sum_{d \in B_{2\rep T}} \sum_{e \in B_T} \Phi(a) p_A(a,b) p_D(b,c) p_A(c,d) p_D(d,e) \\
            &\le O(T) \cdot \sum_{b \in B_T \backslash \{ n \}} \sum_{d \in B_{2\rep T}} e^{2(\rep-1)(b'+1)} \cdot \frac{(4b'(d'+1))^{b'}}{\binom{n}{b}} \\
            &= O(T) \cdot \sum_{1 \le b' \le T} \sum_{d' \le 2\rep T} \frac{(O(T^2))^{b'}}{\binom{n}{b'}} \\
            &= \widetilde{O}(n^{-1}),
        \end{align*}
        where the maximizing term for the second to last expression is $b' = 1$. 

    \item 
        Finally, consider $d \in B_{2\rep T}$ and $b \not\in B_T$. 
        Let us substitute the following bounds:
        \begin{itemize}
        \item 
            First, using \Cref{prop:initial-properties}, \Cref{item:initial-middle-bound} and \Cref{prop:bounds}, \Cref{item:pDpA-trivial-bound}, we have that 
            \[
                \sum_{a=1}^n \Phi(a) p_A(a, b) 
                \le \sum_{a=1}^n \poly(n) \cdot \dbinom{n}{a}^{1-1/\rep} \cdot \min \left\{ 1, \frac{\binom{n}{b}}{\binom{n}{a}} \right\} 
                \le \poly(n) \cdot \dbinom{n}{b}^{1-1/\rep}. 
            \]
        \item 
            Second, \Cref{prop:bounds}, \Cref{item:pDpA-bound} says that $p_D(b,c) p_A(c,d) \le \frac{\left( 4b'(d'+1) \right)^{\lfloor c/2 \rfloor}}{( \lfloor c/2 \rfloor!)^2 \binom{n}{b}}$. 
            We write
            \begin{align*}
                \sum_{c = 1}^n p_D(b, c) p_A(c, d) 
                &\le \frac{1}{\binom{n}{b}} \cdot \sum_{c = 1}^{n} \frac{\left( 4b'(d'+1) \right)^{\lfloor c/2 \rfloor}}{( \lfloor c/2 \rfloor!)^2} \\
                &\le O(1) \cdot \frac{\exp(2\sqrt{4b'(d'+1)})}{\binom{n}{b}} \\
                &\le \frac{n^{O(\sqrt{b'})}}{\binom{n}{b}},
            \end{align*}
            where we used in the second inequality that $\sum_{c=0}^{n/2} \frac{z^{c}}{(c!)^2} \le \left( \sum_{c=0}^\infty \frac{(\sqrt{z})^c}{c!} \right)^2 = e^{2\sqrt{z}}$, and we use in the last inequality that $d' = O(T) = O(\log^2 n)$. 
        \item 
            Finally, we use the trivial bound $p_D(d, e) \le 1$. 
        \end{itemize}
        Combining the above bounds, we get that 
        \begin{align*}
            \sum_{a = 1}^n \sum_{b \not\in B_T} \sum_{c = 1}^n \sum_{d \in B_{2\rep T}} \sum_{e \in B_T} &\Phi(a) p_A(a,b) p_D(b,c) p_A(c,d) p_D(d, e) \\
            &\le \sum_{b \not\in B_T} \sum_{d \in B_{2\rep T}} \poly(n) \cdot \dbinom{n}{b}^{1-1/\rep} \cdot \frac{n^{O(\sqrt{b'})}}{\binom{n}{b}} \cdot 1 \\
            &= \sum_{b \not\in B_T} \sum_{d \in B_{2\rep T}} \frac{n^{O(1 + \sqrt{b'})}}{\binom{n}{b}^{1/\rep}} \\
            &\le n^{-\Omega(T)} \\
            &= \negl(n).
        \end{align*}
    \end{enumerate}
\end{proof}

Now we may bound the total boundary contribution for $P_\rep(AD)^m$ codes.

\begin{theorem} \label{thm:AD-boundary}
    For $m \ge 2$, $\sum_{c \in B_T} W_{AD}^{(m)}(c) = \widetilde{O}(n^{1-m})$. Therefore, for $m \ge 3$, the expected number of codewords of a $P_\rep(AD)^m$ code that have weight $< \tau n$ or $> (1 - \tau)n$ after $m-1$ rounds is $\widetilde{O}(n^{2 - m})$. 
\end{theorem}

\begin{proof}
    The statement for $m = 2$ was given in \Cref{lemma:AD-base-case}. For $m \ge 3$, we can split 
    \[
        W_{AD}^{(m)}(c) 
        = W_{AD, < \tau}^{(m)}(c) + W_{AD, \ge \tau}^{(m)}(c).
    \]
    The latter term was bounded as $\le \negl(n)$ in~\Cref{thm:middle} for $c' \le (\delta^{(m)} - \eps)n$, which holds for us since we care about $c' \le T = \log^2(n)$. As for the first term, its recursion formula is given by
    \[
        W_{AD, < \tau}^{(m)}(c) = \sum_{a \in B_T} W_{AD}^{(m-1)}(a) \cdot \sum_{b \in [n]} p_A(a, b) p_D(b, c) 
    \]
    We split this formula by whether $b \in B_{2\rep T}$ or not. For $b \in B_{2\rep T}$, we have, using \Cref{prop:bounds}, \Cref{item:pA-bound} and the induction hypothesis, that
    \begin{align*}
        \sum_{a \in B_T} W_{AD}^{(m-1)}(a) \sum_{b \in B_{2\rep T}} p_A(a,b) p_D(b,c) 
        &\le \sum_{a \in B_T} W_{AD}^{(m-1)}(a) \sum_{b \in B_{2\rep T}} p_A(a, b) \\
        &\le \widetilde{O}(n^{2-m}) \sum_{b \in B_{2\rep T}} \left( \frac{4(b'+1)}{n} \right)^{\lceil 1/2 \rceil} \\
        &\le \widetilde{O}(n^{2-m}) \cdot \widetilde{O}(n^{-1}) \\
        &= \widetilde{O}(n^{1-m}),
    \end{align*}
    where in the second inequality we use that $a \ge 1$ and in the last inequality we used that $b' \le 2\rep T$. For $b \not\in B_{2\rep T}$, \Cref{prop:bounds}, \Cref{item:pT-trivial-bound} gives that 
    \[
        \sum_{a \in B_T} \sum_{b \not\in B_{2\rep T}} W_{AD}^{(m-1)}(a) p_A(a,b) \cdot p_D(b, c) 
        \le \sum_{b \not\in B_{2\rep T}} \frac{\poly(n) \cdot \binom{n}{c}}{\binom{n}{b}^{1/\rep}}
        \le n^{-\Omega(T)}
        = \negl(n),
    \]
    using $b' > 2\rep T$ and $c' \le T$.
\end{proof}

\section{Human and AI Interaction}

AI tools, specifically ChatGPT 5.6 Sol and 6.0 Astra, played a major part in the development of this work. For one, the linear time encoder was suggested by ChatGPT 5.6 Sol. Further, the author used these tools alongside reading the white papers to understand and generalize the proofs of~\cite{BR26,FagnaniR08}. Interestingly, she found that early attempts to ask ChatGPT how previous work should be generalized were actually erroneous: ChatGPT reported that proofs should remain nearly identical with small polynomial adjustments, which was not in general true as polynomial adjustments could erase the polynomial bounds proved in~\Cref{thm:DA-boundary,thm:AD-boundary}. ChatGPT also seemed to be quick to build upon previously written theorems without verification: the various bugs corrected in this manuscript were only found because the author was reading the original papers and did not understand why certain statements held. Once prompted to correct these bugs, however, ChatGPT provided corrected proofs immediately. 

All proofs and computations were worked out and understood by the author by hand (thanks to the Simons Institute for a vast supply of printer paper). The workflow for developing the proofs was: she first used ChatGPT to run numerical simulations to check plausibility; once plausibility was established, she sketched an outline of the proof using~\cite{BR26} as inspiration, then worked out individual parts of the proof by hand, reading~\cite{BR26} or asking ChatGPT to explain certain bounds. 
The final manuscript was written entirely by the author, though she used ChatGPT to locate minor bugs and typos at the end.

\bibliographystyle{alpha}
\bibliography{refs}

\appendix

\section{Proofs of Properties and Bounds} \label{sec:appendix-properties}

We provide proofs of the various facts and bounds we used throughout the paper. 

\phifacts*

\begin{proof}
\begin{enumerate}
\item 
    There is a bijection between weight $a$ strings and weight $n-a$ strings in $\widetilde{P}_\rep$: simply flip all bits. Because $\rep$ is even, the resulting flipped string is still in $\widetilde{P}_\rep$. 
\item 
    Our starting distribution has no codewords of odd weight.
\item 
    Consider the weight enumerator function for each $\rep$-block. It has the form $w_\rep(x) = \sum_{j = 0}^{\rep/2} \binom{\rep}{2j} x^j$, where the coefficient of $x^j$ is the number of length $\rep$ words of weight $2j$. We bound $\binom{\rep}{2j} = 2^j \cdot \frac{(\rep-1) (\rep-3) \cdots (\rep-2j+1)}{2j (2j-1) \cdots (j+1)} \cdot \binom{\rep/2}{j} \le (\rep-1)^j \cdot \binom{\rep/2}{j}$. Therefore $w_\rep(x) \le \sum_{j=0}^{\rep/2} \binom{\rep/2}{j} \cdot (\rep-1)^j x^j = (1 + (\rep-1) x)^{\rep/2}$. 
    
    The weight enumerator for the whole code is therefore $w_\rep(x)^{n/\rep} \le (1 + (\rep-1)x)^{\rep/2 \cdot n/\rep} = (1 + (\rep-1)x)^{n/2} = \sum_{j = 0}^{n/2} (\rep-1)^j \binom{n/2}{j} x^j$, where the coefficient of $x^j$ is an upper bound on the number of weight $2j$ words. Reading off the $(a/2)$'th coefficient, we get that $\Phi(a) \le (\rep-1)^{a/2} \binom{n/2}{a/2}$. 
\item 
    Consider a random $n$-bit string $X$ sampled by setting each bit independently to $1$ with probability $\alpha$. We have that the probability a given $\rep$-block has even parity is $\frac{1 + (1 - 2\alpha)^\rep}{2}$, so the probability all $n/\rep$ blocks have even weight and thus $X \in \widetilde{P}_\rep$ is $\left( \frac{1 + (1 - 2\alpha)^\rep}{2} \right)^{n/\rep}$. On the other hand, each $x \in \widetilde{P}_\rep$ is sampled in this way with probability $\alpha^{|x|} (1 - \alpha)^{n - |x|}$, so 
    \[
        \Phi(a) \cdot \alpha^a (1 - \alpha)^{n-a} \le \left( \frac{1 + (1 - 2\alpha)^\rep}{2} \right)^{n/\rep}.
    \]
\item 
    To see that this property holds for our distribution $\Phi(a)$, consider a uniformly random weight $a$ string $x \in \widetilde{P}_\rep$. Then since $x$ is uniform, $\bbP[x_i = 1] = a/n$ for all $i \in [n]$. If we delete one coordinate from each $\rep$-block, the remaining $\rep-1$ bits determine the last bit. Thus, $\log_2 \Phi(a) \le H(x) \le \left( 1 - \frac1\rep \right) n h(a/n) \le \left( 1 - \frac1\rep \right) \cdot \log_2 \left( \poly(n) \cdot \binom{n}{a} \right)$. 
\item 
    Verifying $\varphi(\gamma) - h(\gamma)$ is symmetric around $\frac12$ is easy (using that $\rep$ is even). Let us restrict to $\gamma \le \frac12$. Write $A(\gamma) = h(\gamma) + \frac1\rep \cdot \log_2 \left( \frac{1 + (1 - 2\gamma)^\rep}{2} \right)$ and $B(\gamma) = \frac{\log_2(\rep - 1)}{2} \cdot \gamma + \frac12 \cdot h(\gamma)$, so that $\varphi(\gamma) := \min \{ A(\gamma), B(\gamma) \}$. It suffices to verify that both $A(\gamma)$ and $B(\gamma)$ are increasing. 

    $B(\gamma)$ is the easy case: $B(\gamma) = \frac{\log_2(s-1)}{2} \cdot \gamma + \frac12 \cdot h(\gamma)$, both $\frac{\log_2(s-1)}{2} \cdot \gamma$ and $\frac12 \cdot h(\gamma)$ are increasing in $\gamma$. 

    $A(\gamma)$ requires a bit more work. Differentiating, we get 
    \begin{align*}
        A'(\gamma) &= \log_2 \frac{1 - \gamma}{\gamma} - \frac{2(1 - 2\gamma)^{\rep - 1}}{\ln 2 \cdot (1 + (1 - 2\gamma)^\rep)} \\
        &\ge \log_2 \frac{1 - \gamma}{\gamma} - \frac{2(1 - 2\gamma)}{\ln 2}.
    \end{align*}
    Letting $t = 1 - 2\gamma$, we use the inequality $\log_2 \frac{1 - \gamma}{\gamma} = \log_2 \frac{1 + t}{1 - t} = \frac{2}{\ln 2}\int_0^t \frac{du}{1 - u^2} > \frac{2t}{\ln 2} = \frac{2(1 - 2\gamma)}{\ln 2}$. Therefore $A'(\gamma) > 0$. 
    
\item 
    Let us restrict to $\gamma \le \frac12$. It is easy to see that $A(\gamma) - h(\gamma) = \frac1\rep \cdot \log_2 \left( \frac{1 + (1 - 2\gamma)^\rep}{2} \right)$ is decreasing as $\gamma$ increases to $\frac12$. As for $B(\gamma)$, we can take the derivative, and see that $\frac{\partial }{\partial \gamma} \left( B(\gamma) - h(\gamma) \right) = \frac12 \log_2 \left( \frac{(\rep - 1) \gamma}{1 - \gamma} \right)$, which is negative iff $\gamma < \frac1\rep$. Then, whenever $\gamma < \frac1\rep$, $B(\gamma) - h(\gamma)$ is decreasing. When $\gamma \ge \frac1\rep$, we have that $B(\gamma) \ge A(\gamma)$ (proved below), so it is not the relevant branch in the definition of $\varphi(\gamma)$.

    Let's prove that $B(\gamma) \ge A(\gamma)$ for $\gamma \ge \frac1\rep$. It suffices to show that $B(1/\rep) \ge A(1/\rep)$, as then $B(\gamma) - h(\gamma) > B(1/\rep) - h(1/\rep) \ge A(1/\rep) - h(1/\rep) > A(\gamma) - h(\gamma)$ for $\gamma > \frac1\rep$. We can explicitly compute: 
    \[
        \rep \left( A(1/\rep) - B(1/\rep) \right)
        = \log_2 \left( \frac{1 + (1 - 2/\rep)^\rep}{2(1 - 1/\rep)^{\rep/2}} \right),
    \]
    so it suffices to show
    \[
        \frac{1 + (1 - 2/\rep)^\rep}{2} \le \left( 1 - \frac1\rep \right)^{\rep/2}
    \]
    This follows from the following inequality which is the $2$-point $L^2 \rightarrow L^\rep$ specialization of the Bonami-Beckner hypercontractive inequality~\cite{Beckner1975Inequalities}: For $\rep \ge 2$ and $a,b \in \bbR$,
    \[
        \frac{|a+b|^\rep + |a-b|^\rep}{2} \le (a^2 + (\rep - 1) b^2)^{\rep/2}
    \]
    Now the desired inequality follows from plugging in $a = 1 - \frac1\rep$ and $b = \frac1\rep$.

\item 
    Borrowing above notation, we need only prove that $A(\gamma)$ and $B(\gamma)$ are both $\ge h(\gamma) - 1/\rep$. First, we have $A(\gamma) \ge h(\gamma) - 1/\rep$ because $\frac{1 + (1 - 2\gamma)^\rep}{2} \ge \frac12$ so $\log_2 \left( \frac{1 + (1 - 2\gamma)^\rep}{2} \right) \ge -1$. For $B(\gamma)$, it suffices to show for $\gamma \le \frac12$ that $\frac2\rep + \gamma \log_2(\rep - 1) \ge h(\gamma)$. Since $h$ is concave, the tangent line at $\gamma = \frac1\rep$ is an upper bound:
    \begin{align*}
        h(\gamma) 
        &\le h(1/\rep) + h'(1/\rep) \cdot (\gamma - 1/\rep) \\
        &= \log_2 \left( \frac{\rep}{\rep - 1} \right) + \gamma \log_2 ( \rep - 1 ) \\
        &\le \frac2\rep + \gamma \log_2 ( \rep - 1 ),
    \end{align*}
    where $\log_2 \left( \frac{\rep}{\rep - 1} \right) \le \frac2\rep$ holds by Jensen on the convex function $f(t) = - \log_2(1 - t)$: $\log_2 \left( \frac{\rep}{\rep - 1} \right) = f(1/\rep) \le \frac2\rep \cdot f(\frac12) + (1 - \frac2\rep) \cdot f(0) = \frac2\rep$. 
\end{enumerate}
\end{proof}

\fAfacts*

\begin{proof}
We prove each statement below.
\begin{enumerate}
\item 
    We have that 
    \[
        f_A(\alpha, \gamma) 
        = \alpha \cdot h\left( \frac12 \right) - h(\gamma) + (1 - \alpha) \cdot h\left( \frac{\gamma - \alpha/2}{1 - \alpha} \right).
    \]
    By concavity of $h$, we have that $h(\gamma) = h \left( \alpha \cdot \frac12 + (1 - \alpha) \cdot \frac{\gamma - \alpha/2}{1 - \alpha} \right) \ge \alpha \cdot h\left( \frac12 \right) + (1 - \alpha) \cdot h\left( \frac{\gamma - \alpha/2}{1 - \alpha} \right)$, hence $f_A(\alpha, \gamma) \le 0$. Equality is achieved only when $\alpha = 0$ or when $\frac{\gamma - \alpha/2}{1 - \alpha} = \frac12 \Longleftrightarrow \gamma = \frac12$. 
\item 
    This statement is \cite{BR26}, Lemma 4.2(3).
\item 
    This statement is \cite{BR26}, Lemma 4.2(4).
\item 
    This statement is \cite{BR26}, Lemma B.1.
\end{enumerate}
\end{proof}

\rAfacts*

\begin{proof}
We prove each statement below.
\begin{enumerate}
\item 
    $r_A^{(m)}(\gamma) = \max_{\alpha \le \min \{ 2\gamma, 2(1-\gamma) \}} \{ r_A^{(m-1)}(\alpha) + f_A(\alpha, \gamma) \} \ge r_A^{(m-1)}(0) + f_A(0, \gamma) = 0$.
\item 
    The case of $m = 0$ follows from \Cref{prop:initial-properties}, \Cref{item:initial-symmetric,item:initial-phi-increasing}. For $m \ge 1$, let us first show symmetry for $m \ge 1$ by inspecting the definition of $r_A^{(m)}$:
    \[
        r_A^{(m)}(\gamma) = \max_{0 \le \alpha \le \min \{ 2\gamma, 2(1 - \gamma) \}} \left\{ r_A^{(m-1)}(\alpha) + f_A(\alpha, \gamma) \right\}
    \]
    Observe that $f_A(\alpha, \gamma) = \alpha - h(\gamma) + (1-\alpha) h\left( \frac{\gamma - \alpha/2}{1 - \alpha} \right)$ is symmetric about $\gamma = \frac12$ and that the range $0 \le \alpha \le \min \{ 2 \gamma, 2(1 - \gamma) \}$ is the same for $\gamma$ and $1 - \gamma$. 

    The statement that $r_A^{(m)}(\gamma)$ is increasing between $0$ and $\frac12$ follows from the fact that $f_A(\alpha, \gamma)$ is increasing between $0$ and $\frac12$ in $\gamma$ (\Cref{prop:fA-facts}, \Cref{item:fA-increasing-in-gamma}).

\item 
    Let $Q^{(m)}(\gamma) := r_A^{(m)}(\gamma) - h(\gamma)$. Note that $Q^{(m)}(\gamma)$ is symmetric around $\frac12$ because $r_A^{(m)}(\gamma)$ is (\Cref{prop:rA-facts}, \Cref{item:rm-increasing}). Let us restrict to the interval $\gamma \in [0, \frac12]$. The base case for $m = 0$ is given by \Cref{prop:initial-properties}, \Cref{item:initial-phi-h-decreasing}. Now, suppose that the statement holds for $m - 1 \ge 0$. Note that we can write
    \begin{align*}
        Q^{(m)}(\gamma) 
        &= \max_{0 \le \alpha \le \min \{ 2\gamma,2(1-\gamma) \}} \left\{ r_A^{(m-1)}(\alpha) + f_A(\alpha, \gamma) \right\} - h(\gamma) \\
        &= \max_{0 \le \alpha \le \min \{ 2\gamma,2(1-\gamma) \}} \left\{ Q^{(m-1)}(\alpha) + \big( h(\alpha) + f_A(\alpha, \gamma) - h(\gamma) \big) \right\}.
    \end{align*}
    Because $r^{(m-1)}(\alpha)$ is symmetric around $\frac12$ and $f_A(\alpha, \gamma) \ge f_A(1 - \alpha, \gamma)$ when $\alpha \le \frac12$, the maximum must be achieved for $\alpha \le \frac12$. 
    
    Let us write $K(\alpha, \gamma) = h(\alpha) + f_A(\alpha, \gamma) - h(\gamma)$, so that the recursion for $Q$ can be written as 
    \[
        Q^{(m)}(\gamma) = \max_{0 \le \alpha \le \min \{ 2\gamma,2(1-\gamma), \frac12 \}} \left\{ Q^{(m-1)}(\alpha) + K(\alpha, \gamma) \right\}.
    \]
    From \Cref{prop:fA-facts}, \Cref{item:h-stable}, we have that, fixing $\gamma$ and varying $\alpha$, the function $K(\alpha, \gamma)$ is maximized at $\alpha = \lambda_\gamma$ and is decreasing for $\alpha > \lambda_\gamma$. Since by the induction hypothesis $Q^{(m-1)}(\alpha)$ is decreasing for $\alpha \in [0, \min \{ 2\gamma, 2(1 - \gamma) \}] \cap [0, \frac12]$, the maximum must be obtained for $\alpha \in [0, \lambda_\gamma]$. Thus we may restrict the domain of $\alpha$ and write
    \[
        Q^{(m)}(\gamma) 
        = \max_{0 \le \alpha \le \lambda_\gamma} \left\{ Q^{(m-1)}(\alpha) + \big( h(\alpha) + f_A(\alpha, \gamma) - h(\gamma) \big) \right\}.
    \]

    Next, we can compute $\frac{\partial}{\partial \gamma} K(\alpha, \gamma) = \log_2 \frac{\gamma^2 (1 - \gamma - \alpha/2)}{(1-\gamma)^2 (\gamma - \alpha/2)}$. When $\alpha \le \lambda_\gamma$, this derivative is non-positive, so for any $0 \le \gamma_1 \le \gamma_2 \le \frac12$, 
    it holds that 
    \[
        Q^{(m)}(\gamma_1) 
        = \max_{0 \le \alpha \le \lambda_{\gamma_1}} \left\{ Q^{(m-1)}(\alpha) + K(\alpha, \gamma_1) \right\}
        \ge \max_{0 \le \alpha \le \lambda_{\gamma_1}} \left\{ Q^{(m-1)}(\alpha) + K(\alpha, \gamma_2) \right\}.
    \]
    Note that in order to fully compare $Q^{(m)}(\gamma_1)$ to $Q^{(m)}(\gamma_2)$, we need to also consider the value of $Q^{(m-1)}(\alpha) + K(\alpha, \gamma_2)$ for $\alpha \in (\lambda_{\gamma_1}, \lambda_{\gamma_2}]$ and show that this is less than $Q^{(m)}(\gamma_1)$ as well. This can be seen by noting that $Q^{(m-1)}(\alpha)$ is non-increasing by the induction hypothesis and $K(\alpha, \gamma_2) \le 0$ by \Cref{prop:fA-facts}, \Cref{item:h-stable}, so $Q^{(m-1)}(\alpha) + K(\alpha, \gamma_2) \le Q^{(m-1)}(\lambda_{\gamma_1}) + 0 = Q^{(m-1)}(\lambda_{\gamma_1}) + K(\lambda_{\gamma_1}, \gamma_1) \le Q^{(m)}(\gamma_1)$. 
\item 
    Restrict to $\gamma \in [0, \frac12]$. It suffices to show that $r_A^{(m-1)}(\alpha) + f_A(\alpha, \gamma) \le r_A^{(m-1)}(\gamma)$ for all $\alpha \le \min \{ 2\gamma, 2(1-\gamma) \}$. Because $r_A^{(m-1)}(\alpha)$ is symmetric around $\frac12$ and $f_A(\alpha, \gamma)$ is decreasing in $\alpha$, it suffices to restrict to $\alpha \le \frac12$ as well.
    
    When $\alpha \le \gamma$, we have from \Cref{prop:rA-facts}, \Cref{item:rm-increasing} that $r_A^{(m-1)}(\gamma)$ is increasing, so 
    \[
        r_A^{(m-1)}(\alpha) + f_A(\alpha, \gamma) 
        \le r_A^{(m-1)}(\gamma) + f_A(\alpha, \gamma) 
        \le r_A^{(m-1)}(\gamma),
    \]
    where we also used that $f_A(\alpha, \gamma) \le 0$ (\Cref{prop:fA-facts}, \Cref{item:fA-zeros}). When $\alpha \ge \gamma$, we use that $r_A^{(m-1)}(\gamma) - h(\gamma)$ is decreasing from \Cref{prop:rA-facts}, \Cref{item:rm-h-decreasing} to obtain $r_A^{(m-1)}(\alpha) \le r_A^{(m-1)}(\gamma) + h(\alpha) - h(\gamma)$, from which we can get
    \[
        r_A^{(m-1)}(\alpha) + f_A(\alpha, \gamma) 
        \le r_A^{(m-1)}(\gamma) + h(\alpha) - h(\gamma) + f_A(\alpha, \gamma) 
        \le r_A^{(m-1)}(\gamma),
    \]
    using that $h(\alpha) - h(\gamma) + f_A(\alpha, \gamma) \le 0$ from \Cref{prop:fA-facts}, \Cref{item:h-stable}.
\end{enumerate}
\end{proof}

\pAbounds*

\begin{proof}
We prove each statement individually.
\begin{enumerate}
\item 
    Corollary 3.4.2 in~\cite{Pfister03} says that 
    \begin{align*}
        p_A(a,b) 
        &\le \frac{\lceil a/2 \rceil}{b} \cdot 2^a \cdot \left( \frac{b}{n} \right)^{\lceil a/2 \rceil} \left( \frac{n-b}{n} \right)^{\lfloor a/2 \rfloor}.
    \end{align*}
    When $b' = b \le \frac{n}{2}$, recall that $p_A(a,b)$ is only nonzero when $a \le 2b$. Then $\frac{\lceil a/2 \rceil}{b} \le 1$ and $\left( \frac{n-b}{n} \right)^{\lfloor a/2 \rfloor} \le 1$ gives the desired bound. When $b > \frac{n}{2}$ so that $b' = n - b$, we use that $\lfloor a/2 \rfloor \le n - b$ to get
    \begin{align*}
        p_A(a, b) &\le \frac{n - b + 1}{b} \cdot 2^a \cdot \left( \frac{b}{n} \right)^{\lceil a/2 \rceil} \left( \frac{n-b}{n} \right)^{\lfloor a/2 \rfloor} \\
        &\le \frac{n - b + 1}{n} \cdot 2^a \cdot \left( \frac{b}{n} \right)^{\lceil a/2 \rceil - 1} \left( \frac{n-b}{n} \right)^{\lfloor a/2 \rfloor} \\
        &\le 2^a \cdot \left( \frac{n - b + 1}{n} \right)^{\lfloor a/2 \rfloor + 1} \\
        &\le \left( \frac{4(b' + 1)}{n} \right)^{\lceil a/2 \rceil}.
    \end{align*}
\item 
    If $b$ is even, we have that $\binom{a-1}{\lceil b/2 \rceil -1} = \binom{a}{\lfloor b/2 \rfloor} \cdot \frac{b/2}{a} \le \binom{a}{\lfloor b/2 \rfloor}$, using that $p_D(a, b)$ is only nonzero for $b \le 2a$. If $b$ is odd, we have that $\binom{a-1}{\lceil b/2 \rceil - 1} = \binom{a}{\lfloor b/2 \rfloor} \cdot \frac{a - \lfloor b/2 \rfloor}{a} \le \binom{a}{\lfloor b/2 \rfloor}$. Either way, we obtain that 
    \begin{align*}
        p_D(a, b)
        = \frac{\binom{n-a}{\lfloor b/2 \rfloor} \binom{a-1}{\lceil b/2 \rceil - 1}}{\binom{n}{a}} 
        \le \frac{\binom{n-a}{\lfloor b/2 \rfloor} \binom{a}{\lfloor b/2 \rfloor}}{\binom{n}{a}} 
        \le \frac{(a'n)^{\lfloor b/2 \rfloor}}{\left( \lfloor b/2 \rfloor! \right)^2 \binom{n}{a}}.
    \end{align*}    
\item 
    This bound follows from combining \Cref{item:pA-bound,item:pD-bound2}, and by using $p_A(b, c) \le 1$:
    \begin{align*}
        p_D(a, b) \cdot p_A(b, c) 
        &\le \frac{(a' n)^{\lfloor b/2 \rfloor}}{(\lfloor b/2 \rfloor!)^2 \binom{n}{a}} \cdot \min \left\{ 1,~\left( \frac{4(c'+1)}{n} \right)^{\lceil b/2 \rceil} \right\} \\
        &\le \frac{(4a'(c'+1))^{\lfloor b/2 \rfloor}}{(\lfloor b/2 \rfloor!)^2 \binom{n}{a}} \\
        &\le \frac{(4a'(c'+1))^{\lfloor b/2 \rfloor}}{\binom{n}{a}}.
    \end{align*}
\item 
    This bound is a trivial bound following from the invertibility of $A$ and $D$: there are only $\binom{n}{b}$ possible outputs, so at most $\binom{n}{b}$ possible weight $a$ strings could have mapped to them. 
\item 
    First, we claim that $\Phi(a_0) \prod_{i=1}^{k} p_{T_i} (a_{i-1}, a_i) \le \poly(n) \binom{n}{a_{k}}^{1 - 1/\rep}$ for $k \ge 0$. We proceed by induction. The case for $k = 0$ is \Cref{prop:initial-properties}, \Cref{item:initial-middle-bound}. Suppose that for $k-1$ it holds that $\Phi(a_0) \prod_{i=1}^{k-1} p_{T_i} (a_{i-1}, a_i) \le \poly(n) \binom{n}{a_{k-1}}^{1 - 1/\rep}$. Then 
    \[
        \Phi(a_0) \prod_{i=1}^{k} p_{T_i}(a_{i-1}, a_i) 
        \le \poly(n) \dbinom{n}{a_{k-1}}^{1 - 1/\rep} \cdot \min \left\{ 1, \frac{\binom{n}{a_k}}{\binom{n}{a_{k-1}}} \right\}
        \le \poly(n) \dbinom{n}{a_k}^{1-1/\rep}.
    \]
    Next, we combine the above bound for $k-1$ with the trivial bound to get
    \begin{align*}
        \Phi(a_0) \prod_{i=1}^{k} p_{T_i} (a_{i-1}, a_i) 
        &\le \poly(n) \binom{n}{a_{k-1}}^{1 - 1/\rep} \cdot p_{T_k} (a_{k-1}, a_k) \\
        &\le \poly(n) \binom{n}{a_{k-1}}^{1 - 1/\rep}  \cdot \frac{\binom{n}{a_k}}{\binom{n}{a_{k-1}}} \\
        &= \frac{\poly(n) \cdot \binom{n}{a_k}}{\binom{n}{a_{k-1}}^{1/\rep}}
    \end{align*}
\end{enumerate}
\end{proof}

\end{document}